\documentclass[oribibl]{llncs}
\usepackage[ansinew]{inputenc}
\usepackage[T1]{fontenc}

\usepackage{amsmath} 
\usepackage{amssymb}

\newcommand{\set}[1]{\{ #1 \} }

\def\nat{{\mathbb N}}
 \def\real{{\mathbb R}}

\DeclareMathOperator{\sep}{sep}

\begin{document}

\title{Upper bounds for Newton's method on\\monotone
polynomial systems,\\ and  
P-time 
model checking of\\probabilistic one-counter automata\thanks{A full
version of this paper is available at {\tt arxiv.org/abs/1302.3741}. Research
partially supported by NSF Grant CCF-1017955.}} 
\author{Alistair Stewart\inst{1} \and Kousha Etessami\inst{1} \and 
Mihalis Yannakakis\inst{2}}
\institute{School of Informatics, University of Edinburgh\\{\tt kousha@inf.ed.ac.uk} , {\tt stewart.al@gmail.com}
\and 
Department of Computer Science, Columbia University\\
{\tt mihalis@cs.columbia.edu}}

\maketitle

\thispagestyle{empty}

\begin{abstract}
A central computational problem for analyzing and model checking 
various classes of infinite-state
recursive probabilistic systems (including 
quasi-birth-death processes, 
multi-type branching processes, 
stochastic context-free
grammars, probabilistic pushdown automata and recursive Markov chains) 
is the computation of {\em termination probabilities}, and  computing these probabilities
in turn boils down to computing the {\em least fixed point} (LFP) solution of a corresponding
{\em monotone polynomial system} (MPS) of equations, denoted $x=P(x)$.

It was shown by Etessami and Yannakakis \cite{rmc} that 
a decomposed variant of Newton's method converges monotonically to 
the LFP solution for any MPS that has a non-negative solution. Subsequently,
Esparza, Kiefer, and Luttenberger \cite{lfppoly} obtained 
upper bounds on the convergence rate of Newton's method for certain classes of MPSs.
More recently, better upper bounds 
have been obtained for special classes
of MPSs (\cite{EWY-qbd-2010,ESY12}). 

However, prior to this paper,  for arbitrary (not necessarily strongly-connected) MPSs,
no upper bounds at all were known 
on the convergence rate of Newton's method
as a function of the encoding
size $|P|$ of the input MPS, $x=P(x)$.

In this paper 
we provide worst-case upper bounds, as a function of both the
input encoding size $|P|$, and $\epsilon > 0$, on the number of iterations
required for decomposed Newton's method ({\em even with rounding}) to converge
to within additive error $\epsilon > 0$ of $q^*$,
for an {\em arbitrary} MPS with LFP solution $q^*$.  
Our upper bounds are essentially optimal in terms of several important
parameters of the problem.

Using our upper bounds, and building on prior work,
we obtain the first P-time algorithm (in the standard Turing model of computation)
for quantitative model checking, to within arbitrary desired precision,
of discrete-time QBDs and
(equivalently) probabilistic 1-counter automata,  with respect to any 
(fixed) $\omega$-regular or LTL property.
\end{abstract}

\section{Introduction}

In recent years, there has been extensive work on the analysis of 
various classes of infinite-state recursive probabilistic systems,
including recursive Markov chains, 
probabilistic pushdown systems,
stochastic context-free grammars, multi-type branching processes,
quasi-birth-death processes and 
probabilistic 1-counter automata 
(e.g. \cite{rmc,EY-MC-12,EKM,ESY12,EWY-qbd-2010,BKK11}).
These are all finitely-presentable models that specify an
infinite-state underlying probabilistic system.
These classes of systems arise in a variety of fields and
have been studied by various communities.
Recursive Markov chains (RMC), and the equivalent model of probabilistic pushdown systems (pPDS), are natural models for probabilistic programs with
recursive procedures \cite{rmc,EKM}. 
Quasi-birth-death (QBD) processes, which are essentially
equivalent (in discrete-time) to probabilistic 1-counter automata (p1CA), are used
in queueing theory and performance evaluation \cite{Neuts81,LatRam99}.  
Stochastic context-free grammars are a central model
in natural language processing and are used also in biology 
\cite{DEKM99}, and branching processes are a classical probabilistic model
with many applications, including in population genetics
(\cite{harris63}).

A central problem for the analysis and model checking of these systems
is the computation of their associated {\em termination probabilities}.
Computing these probabilities amounts to solving a system of fixed-point
multivariate equations $x=P(x)$, where $x$ is a (finite) vector
of variables and $P$ is a vector of polynomials with positive coefficients;
such a system of equations is called a {\em monotone polynomial system} (MPS) 
because $P$ defines a monotone operator from the nonnegative orthant to itself. Each of the above classes has the property that, 
given a model $M$ in the class,
we can construct in polynomial time a corresponding MPS $x=P(x)$ such that
the termination probabilities of $M$ (for various initial states)
are the {\em least fixed point} (LFP) solution of the system, 
i.e., they satisfy the system, and any other nonnegative solution
is at least as large in every coordinate.
In general, a monotone polynomial system may not have any fixed point;
consider for example $x=x+1$. However, if it has a fixed point, then it
has a least fixed point (LFP). The systems constructed from probabilistic
systems as above always have a LFP, which has values in $[0,1]$ 
since its coordinates give the termination probabilities.

The equations are in general nonlinear, and their LFP solution
(the vector of termination probabilities) is in general irrational
even when all the coefficients of the polynomials (and the numerical
input data of the given probabilistic model) are rational.
Hence we seek to compute the desired quantities
up to a desired accuracy $\epsilon >0$. 
The goal is to compute them as efficiently as possible,
as a function of the encoding size of the input (the given probabilistic
model, or the MPS) and the accuracy $\epsilon$. 
We first review some of the relevant previous work and
then describe our results.

{\bf Previous Work.}
An algorithm for computing the LFP of MPSs, based on Newton's method, 
was proposed in \cite{rmc}.
Given a MPS, we can first identify in
polynomial time the variables that have value 0 in the LFP
and remove them from the system, yielding a new so-called {\em cleaned}
system. Then a dependency graph between the variables is constructed,
the variables and the MPS are decomposed into strongly connected
components (SCCs), and Newton's method is applied bottom-up 
on the SCCs, starting from the all-0 vector. It was shown
in \cite{rmc} that, for any MPS that has a (nonnegative) solution,
the decomposed variant of Newton's method
converges monotonically to the LFP. Optimized variants
of decomposed Newton's method have by now been implemented in
several tools (see, e.g., \cite{WojEte07,NedSat08}), 
and they perform quite well in practice on many instances.

Esparza, Kiefer and Luttenberger studied in detail the rate of
convergence of Newton's method on MPSs \cite{lfppoly} (with or without 
decomposition). On the negative side, they showed that there are instances
of MPSs $x=P(x)$ (in fact even simple RMCs), with $n$ variables,
where it takes an exponential number of iterations in the input size to get even within just one bit of precision
(i.e. accuracy 1/2). On the positive side, they showed that after some
initial number $k_P$ of iterations in a first phase, 
Newton's method
thereafter gains bits of precision at a linear rate, meaning
that $k_P + c_P \cdot i$ iterations suffice to gain $i$ bits of precision,
where both $k_P$ and $c_P$ depend on the input, $x=P(x)$. 
For strongly connected MPSs,
they showed that the length, $k_P$, of the initial phase is upper bounded
by an exponential function of the input size $|P|$,
and that $c_P = 1$. For general MPSs
that are not strongly connected (and for general RMCs and pPDSs),
they showed that $c_P = n 2^n$ suffices, but they provided
no upper bound at all on $k_P$ 
(and none was known prior to the present paper).
Thus, they obtained no upper bounds,
as a function of the size of the input, $x=P(x)$, for the number of iterations 
required to get to within even the first bit of precision 
(e.g., to estimate within $< 1/2$ the termination probability of a RMC)
for general MPSs and RMCs.
Proving such a general bound was left as an open problem in \cite{lfppoly}.

For special classes of probabilistic models (and MPSs) better results
are now known. For the class of quasi-birth-death processes (QBDs) and
the equivalent class of probabilistic 1-counter automata (p1CA),
it was shown in \cite{EWY-qbd-2010} that the decomposed Newton method
converges in a polynomial number of iterations in the size of the input
and the bits of precision, and hence the desired termination
probabilities of a given p1CA $M$ can be computed within 
absolute error $\epsilon = 2^{-i}$ in a number of arithmetic operations
that is polynomial in the size $|M|$ of the input and the
number $i=\log(1/\epsilon)$ of bits of precision.
Note that this is {\em not} polynomial time in the standard Turing model
of complexity, because the numbers that result from 
the arithmetic operations in general can become exponentially long
(consider $n$ successive squarings of a number). Thus, the result of
\cite{EWY-qbd-2010} shows that the termination problem for p1CAs can be solved in
polynomial time in the {\em unit-cost exact rational arithmetic} model,
a model in which arithmetic operations cost 1 time unit, regardless
of how long the numbers are. It is not known exactly how powerful
the unit-cost rational model is, but it is believed to be strictly
more powerful than the ordinary Turing model.
The question whether the termination probabilities of a p1CA (and a QBD)
can be computed in polynomial time (in the standard model) was left 
open in \cite{EWY-qbd-2010}.

Building on the 
results of \cite{EWY-qbd-2010} for computation of termination probabilities
of p1CAs,  more recently Brazdil, Kiefer and Kucera \cite{BKK11} 
showed how to do quantitative 
model checking of $\omega$-regular properties (given by
a deterministic Rabin automaton) for p1CAs, i.e., compute within desired 
precision $\epsilon > 0$ the probability that a run of a given p1CA, $M$,
is accepted by a given deterministic Rabin automaton, $R$, 
in time polynomial in
$M, R, \log(1/\epsilon)$ in the unit-cost rational arithmetic model.
The complexity
in the standard Turing model was left open.

For the classes of stochastic context-free grammars, multi-type
branching processes, and the related class of 1-exit RMCs,
we showed recently in \cite{ESY12} that termination probabilities can
be computed to within precision $\epsilon$ in polynomial time
in the size of the input model and $\log(1/\epsilon)$ 
(i.e. the $\#$ of bits of precision) in the standard Turing model \cite{ESY12}.
The algorithm is a variant of Newton's method, where the preprocessing
identifies and eliminates (in P-time \cite{rmc}) 
the variables that have value 1 
in the LFP (besides the ones with value 0). Importantly,
the numbers throughout the computation are not allowed to grow
exponentially in length, but are always rounded down to a polynomial
number of bits. The analysis then shows that the rounded Newton's algorithm
still converges to the correct values (the LFP) and the number of 
iterations and the entire time complexity is polynomially bounded.

For general RMCs (and pPDSs) and furthermore for general MPSs,
even if the LFP is in $[0,1]^n$, 
there are negative results indicating that it is probably
impossible to compute the termination probabilities and the LFP
in polynomial time in the standard Turing model.
In particular, we showed in \cite{rmc} that approximating
the termination probability of a RMC within {\em any} constant
additive error $< \frac{1}{2}$,
is at least as hard as the {\em square-root-sum} problem,
a longstanding open problem that
arises often in computational geometry, which is not even known to be in NP,
and that it is also as hard as the more powerful problem, called PosSLP \cite{ABKM06},
which captures the essence of unit-cost rational arithmetic.
Thus, if one can approximate the termination probability of
a RMC in polynomial time then it is possible to simulate unit-cost
rational arithmetic in polynomial time in the standard model,
something which is highly unlikely.

As we mentioned at the beginning, 
computing termination probabilities is
a key ingredient for performing other,  more general analyses,
including model checking 
\cite{EY-MC-12,EKM}.

\medskip

{\bf Our Results}.
We provide a thorough analysis of decomposed Newton's method
and show upper bounds on its rate of convergence as a function 
of the input size and the desired precision, 
which holds for {\em arbitrary} monotone polynomial systems.
Furthermore, we analyze a {\em rounded} version of the algorithm where
the results along the way are not computed exactly to arbitrary
precision but are rounded to a suitable number of bits
(proportional to the number of iterations $k$ of Newton's method
 that are performed),
while ensuring that the algorithm stays well-defined
and converges to the LFP. Thus, the bounds we show hold
for the standard Turing model and not only the unit-cost model.
Note that all the previous results on Newton's method that we mentioned,
except for \cite{ESY12},
assume that the computations are carried out in {\em exact} arithmetic.
To carry out $k$ iterations of Newton's method with exact arithmetic 
can require exponentially many bits, as a function of $k$,
to represent
the iterates.
In general, the fact that Newton's method converges with exact arithmetic
does not even imply automatically that rounded
Newton iterations will get anywhere close to the solution
when we round to, say, only polynomially many bits of precision
as a function of the number of iterations $k$, 
let alone that the same bounds on the convergence rate will continue to hold.
We nevertheless show that suitable rounding works for MPSs.

In more detail, suppose that the given (cleaned) MPS $x=P(x)$
has a LFP $q^*>0$.  The decomposition into
strongly connected components yields a DAG of SCCs with depth $d$,
and we wish to compute the LFP with (absolute) error 
at most $\epsilon$. Let $q^*_{\min}$ and $q^*_{\max}$ be the minimum
and maximum coordinate of $q^*$.
Then the rounded decomposed Newton method will converge
to a vector $\tilde{q}$
within $\epsilon$ of the LFP, i.e., such that $\|q^* - \tilde{q} \|_\infty \leq \epsilon$ in time polynomial in
the size $|P|$ of the input, $\log(1/\epsilon)$,
$\log(1/q^*_{\min})$, $\log( q^*_{\max})$, and $2^d$
(the depth $d$ in the exponent can be replaced by the
maximum number of {\em nonlinear} SCCs in any path of the DAG of SCCs).
We also obtain bounds on $q^*_{\min}$ and $q^*_{\max}$
in terms of $|P|$ and the number of variables $n$,
so the overall time needed is polynomial in $|P|$, $2^n$ and $\log(1/\epsilon)$.
We provide actually concrete expressions on the number of iterations
and the number of bits needed.
As we shall explain, the bounds are essentially optimal in terms of several
parameters.
The analysis is quite involved and
builds on the previous work. 
It uses several results and techniques
from \cite{rmc,lfppoly,ESY12}, 
and develops substantial additional machinery.

We apply our results then to probabilistic 1-counter automata (p1CAs).
Using our analysis for the rounded decomposed Newton method and 
properties of p1CAs from \cite{EWY-qbd-2010}, 
we show that termination probabilities of a p1CA $M$ (and QBDs)
can be computed to desired precision $\epsilon$ in polynomial time
in the size $|M|$ of the p1CA and $\log(1/\epsilon)$ (the bits of precision) in the standard Turing model of computation,
thus solving the open problem of \cite{EWY-qbd-2010}.

Furthermore, combining with the results of  \cite{BKK11} and
\cite{EY-MC-12}, we show that one can do quantitative model
checking of $\omega$-regular properties for p1CAs in polynomial time
in the standard Turing model,
i.e., we can compute to desired precision $\epsilon$
the probability that a run of a given p1CA $M$ satisfies an
$\omega$-regular property in time polynomial in 
$|M|$ and $\log(1/\epsilon)$ (and exponential in
the property if it is given for example as a non-deterministic B\"{u}chi 
automaton
or polynomial if it is given as a deterministic Rabin automaton).

The rest of the paper is organized as follows.
In Section 2 we give basic definitions and background.
In Section 3 we consider strongly-connected MPS, and in Section 4
general MPS. Section 5 analyzes p1CAs.  Most proofs are in
the appendix.

\section{Definitions and Background}

We first recall basic definitions about MPSs from \cite{rmc}.
A {\em monotone polynomial system of equations} (MPS) consists of 
a system of $n$ equations in $n$  variables,  $x= (x_1, \ldots, x_n)$,
the equations are of the form $x_i = P_i(x)$,  $i= 1,\ldots,n$,   such that
$P_i(x)$ is a multivariate polynomial in the variables $x$,
and such that the monomial coefficients and constant term of $P_i(x)$ are all non-negative.
More precisely, for $\alpha = (\alpha_1, \alpha_2, \ldots, \alpha_n) \in \nat^n$,
we use the notation $x^{\alpha}$ to denote the monomial $x_1^{\alpha_1} x_2^{\alpha_2} \ldots x_n^{\alpha_n}$.
(Note that by definition $x^{(0,\ldots,0)} = 1$.)
Then for each polynomial $P_i(x)$, $i=1,\ldots,n$, there is some {\em finite} subset of $\nat^n$,  denoted ${\mathcal C}_i$,
and for each $\alpha \in {\mathcal C}_i$, there is a positive (rational) coefficient $c_{i,\alpha} > 0$,
such that $P_i(x) \equiv \sum_{\alpha \in {\mathcal C}_i} c_{i,\alpha} x^{\alpha}$.

For computational purposes, we assume each polynomial $P_i(x)$ 
has rational coefficients\footnote{although we also reason
about MPSs with positive real-valued coefficients in our proofs.}, and that it is encoded
succinctly
by specifying the list of pairs $\langle (c_{i,\alpha}, \alpha) \mid \alpha \in {\mathcal C_i} \rangle$,
where
each rational coefficient $c_{i,\alpha}$ is represented by giving its numerator and
denominator in binary,
and each integer vector $\alpha$ is represented in {\em sparse representation},
by  only listing  its non-zero coordinates, $i_1, \ldots, i_k$, by using a list
$\langle (i_1, \alpha)_{i_1}), \ldots, (i_{k}, \alpha_{i_{k}}) \rangle$,
giving each integer $\alpha_{i_j}$ in binary.
(Proposition \ref{prop:snf-form} below, from \cite{rmc,ESY12}, shows that
 using such a sparse representation 
 does not entail any extra computational cost.)

We use vector notation, using $x=P(x)$ to denote the entire MPS.
We use $|P|$ to denote the encoding size (in bits) of the MPS $x=P(x)$
having rational coefficients, using the succinct representation just described.

Let $\real_{\geq 0}$ denote the non-negative real numbers.
Then $P(x)$ defines a monotone operator
on the non-negative orthant $\real^n_{\geq 0}$.  In other words, 
$P: \real^n_{\geq 0} \rightarrow \real^n_{\geq 0}$, and 
if ${\mathbf 0} \leq a \leq b$,  then   $P(a) \leq P(b)$.
In general, an MPS  need not have any real-valued solution: consider $x= x+1$.
However, because of monotonicity of $P(x)$, if there exists a solution $a \in \real^n_{\geq 0}$
such that $a=P(a)$, then there exists a {\em least fixed point} (LFP) solution $q^* \in \real^n_{\geq 0}$
such that $q^* = P(q^*)$, and such that $q^* \leq a$  for all
solutions $a \in \real^{n}_{\geq 0}$.    Indeed, if for $z \in \real^n$ 
we define $P^0(z) = z$, and define $P^{k+1}(z) = P(P^k(z))$, for all $k \geq 0$, then
(as shown in \cite{rmc}) value iteration starting at the all-$0$ vector ${\mathbf 0}$ converges
monotonically to $q^*$: in other words $\forall k \geq 0$  $P^{k}({\mathbf 0}) \leq P^{k+1}({\mathbf 0})$,
and $\lim_{k \rightarrow \infty} P^k({\mathbf 0}) = q^*$.\footnote{Indeed, even if an MPS does not have a 
{\em finite} LFP solution $q^* \in \real^n_{\geq 0}$, it always does have
an LFP solution
{\em over the extended non-negative reals}.  Namely,
we can define the LFP of any MPS, $x=P(x)$, to be the vector  $q^* \in \overline{\real}^n_{\geq 0}$ 
over $\overline{\real}_{\geq 0} = (\real_{\geq 0} \cup \{ + \infty \})$,
given by $q^* := \lim_{k \rightarrow \infty} P^k({\mathbf 0})$. 
In general, it is PosSLP-hard to decide whether a given MPS has a finite LFP.
(This follows easily from results in \cite{rmc}, although it is not stated there: is 
 was shown there it is PosSLP-hard to decide if $q^*_1 \geq 1$ in an MPS with
 finite LFP $q^* \in \real^n_{\geq 0}$.
 Then just add a variable $x_0$, and an equation $x_0 = x_0 x_1 + 1$ to the MPS.  
  In the new MPS,  $q^*_0 = +\infty$ if and only if $q^*_1 \geq 1$.)
However, various classes of MPSs, including those whose LFP corresponds to termination probabilities 
of various recursive probabilistic systems 
do have a finite LFP.
Thus in this paper we will 
only consider LFP computation for MPSs that have a finite LFP $q^* \in \real^n_{\geq 0}$.
So when we say  ``$x=P(x)$ is an MPS with LFP solution $q^*$'', we mean 
$q^* \in \real^n_{\geq 0}$, unless specified otherwise.  \label{foot:extended}}
 
Unfortunately, standard value iteration $P^k(0)$, $k \rightarrow \infty$, can converge very slowly to $q^*$, even for a fixed MPS
with 1 variable, even when $q^*=1$; specifically, $x = (1/2) x^2  + 1/2$ already
exhibits exponentially slow convergence to its LFP $q^* = 1$  (\cite{rmc}).
It was shown in \cite{rmc} that a decomposed variant of 
Newton's method also converges 
monotonically to $q^*$ for an MPS with LFP solution $q^*$.
More recently, in \cite{ESY12}, a version of Newton's method with suitable rounding between iterations
was studied.  Rounding is necessary if one wishes to consider the complexity of Newton's method
in the standard (Turing) model of computation, which does not allow unit-cost arithmetic
operations on arbitrarily large numbers.
In this paper we will apply a version of Newton's method to MPSs 
which uses both rounding and decomposition.  Before describing it, we need some further background.

An MPS, $x = P(x)$, is said to be in
{\em simple normal form} (SNF)
if for every $i = 1,\ldots,n$,  the polynomial $P_i(x)$ has one of two forms:
(1) \underline{Form$_{*}$}: $P_i(x) \equiv x_j x_k$ 
is simply a quadratic monomial;  \ or  (2) \underline{Form$_{+}$}: $P_i(x)$ is a {\em linear} 
expression 
$\sum_{j \in \mathcal{C}_i} p_{i,j} x_j + p_{i,0}$,  for some rational 
non-negative
coefficients $p_{i,j}$ and $p_{i,0}$, and
some index set $\mathcal{C}_i \subseteq
\{1,\ldots,n\}$.
In particular,  in any MPS in SNF form
every polynomial $P_i(x)$ has multivariate degree bounded by at most $2$ in the variables $x$. 
We will call such MPSs {\bf quadratic MPSs}.

As shown in \cite{rmc,ESY12}, it is easy to convert any MPS to SNF form, by adding auxiliary
variables and equations:

\begin{proposition}{(Propos. 7.3 \cite{rmc}, and Propos. 2.1 of \cite{ESY12})}
\label{prop:snf-form}
Every MPS, $x = P(x)$, with LFP $q^*$,  can be transformed in P-time
to an ``{\em equivalent}'' quadratic MPS  $y=Q(y)$  in SNF form,
such that 
$|Q| \in O( |P|  )$.
More precisely, the variables $x$ are 
a subset of the variables $y$, and  $y=Q(y)$ has LFP
$p^*$ iff $x=P(x)$ has LFP $q^*$,
and projecting $p^*$ onto the $x$ variables yields 
$q^*$.
\end{proposition}

Furthermore, for any MPS, $x=P(x)$, 
we can in P-time find and remove
any variables $x_i$, such that the LFP solution has $q^*_i = 0$.$^{\ref{foot:second-extended}}$

\begin{proposition}{(Proposition 7.4 of \cite{rmc})}
There is a P-time algorithm that, given any MPS\footnote{This proposition holds regardless whether
the LFP $q^*$ is {\em finite} or is over the {\em extended non-negative reals},
$\overline{\real}_{\geq 0}$. Such an extended LFP exists for any MPS. See footnote 
\ref{foot:extended}.\label{foot:second-extended}}, $x = P (x)$, 
over $n$ variables, determines for each $i \in \{1, \ldots, n\}$
whether $q^*_i = 0$.
\label{prop:zero-removal-in-p-time}
\end{proposition} 

Thus, for every MPS, we can detect in P-time all the variables 
$x_j$ such that $q^*_j = 0$, remove their equation $x_j = P_j(x)$,
and set the variable $x_j$ to $0$ on the RHS of the remaining
equations.  
We obtain as a result a {\bf\em cleaned} MPS, $x'=Q(x')$,  
which has an LFP $q^* > 0$.

Applying Propositions  
\ref{prop:snf-form} and 
\ref{prop:zero-removal-in-p-time}, 
{\bf\em we assume wlog in the rest of this paper that every MPS is a cleaned
quadratic MPS,  with LFP $q^* > 0$.}\footnote{For compatibility when quoting prior work, 
it will sometimes be convenient to assume 
quadratic MPSs, rather than the more restricted SNF form MPSs.}

In order to describe {\em decomposed} Newton's method,
for a {\em cleaned} MPS, $x=P(x)$
we need to define the {\em dependency graph}, 
$G_P = (V,E)$, of the MPS.
The nodes $V$ of $G_P$ are
the remaining variables $x_i$, 
and the edges are defined as follows: $(x_i, x_j) \in E$ if and only if 
$x_j$ appears in some monomial in $P_i(x)$ that has a positive coefficient.

We shall decompose the cleaned system of equation $x=P(x)$,
into  strongly connected components (SCCs),
using the dependency graph $G_P$ of variables, 
and we shall apply Newton's method separately on each SCC  ``bottom-up''.

We first recall basic definitions for (a rounded down version of) 
Newton's method applied to MPSs.
For an MPS, $x=P(x)$, with $n$ variables, we define $B(x) = P'(x)$ to be the $n \times n$ Jacobian 
matrix of partial derivatives of $P(x)$. 
In other words,  $B(x)_{i,j} =  \frac{\partial P_i(x)}{\partial x_j}$.
For a vector $z \in \real^n$, 
assuming that the matrix $(I - B(z))$ is non-singular, 
a single iteration of {\em Newton's method} ({\bf NM}) on $x=P(x)$ at $z$ 
is defined via the following operator:
\begin{equation}\label{newton-one-it-eq}
\mathcal{N}_P(z) :=     z + (I-B(z))^{-1} (P(z) - z)
\end{equation}

\noindent Let us now recall from \cite{ESY12} 
the rounded down Newton's method, with parameter $h$, applied to an MPS:

\begin{definition} {{\bf Rounded-down Newton's method} ({\bf R-NM}) , with rounding parameter $h$.)}
Given an MPS, $x=P(x)$,  
with LFP $q^*$, 
where ${\textbf 0} < q^*$,
in the {\em rounded down Newton's method} (R-NM) with integer 
rounding parameter $h > 0$, 
we compute a sequence of 
iteration vectors $x^{[k]}$,  where the initial starting vector is $x^{[0]} 
:= \mathbf{0}$,
and such that for each $k \geq 0$, given $x^{[k]}$, we compute 
$x^{[k+1]}$ as follows:

\begin{enumerate}
\item  First, compute $x^{\{k+1\}} :=  \mathcal{N}_P(x^{[k]})$, 
where the Newton
iteration operator $\mathcal{N}_P(x)$ was defined in equation 
(\ref{newton-one-it-eq}).
(Of course we need to show that all such Newton iterations are defined.)

\vspace*{-0.06in}

\item For each coordinate $i=1,\ldots,n$, set $x^{[k+1]}_i$ 
to be equal 
to the maximum (non-negative) multiple of $2^{-h}$ which is $\leq \max (x^{\{k+1\}}_i, 0)$.
(In other words, round down $x^{\{ k+1\}}$ to the nearest multiple of 
$2^{-h}$, while making sure that the result is non-negative.)
\end{enumerate}
\end{definition}

Now we describe the {\bf\em Rounded-down Decomposed Newton's Method} ({\bf R-DNM}) applied to an MPS, $x=P(x)$,
with real-valued LFP $q^* \geq 0$.
Firstly, we use  Proposition \ref{prop:zero-removal-in-p-time} to remove $0$ variables, and thus we
can assume we are given a cleaned MPS, $x=P(x)$, with real-valued LFP $q^* > 0$.
 
Let $H_P$ be the DAG of SCC's of the dependency graph $G_P$.    
We work bottom-up in $H_P$, starting at bottom SCCs. For each SCC, $S$,
suppose its corresponding equations are $x_S = P_S(x_S,x_{D(S)})$, where
$D(S)$ denotes the union of the variables in ``lower'' SCCs, below $S$,
on which $S$ depends.   In other words,  a variable $x_j \in D(S)$ iff there is
some variable $x_i \in S$ such that there is directed
path in $G_p$ from $x_i$ to $x_j$.
If the system $x_S = P_S(x_S,q^*_{D(S)})$ is a linear system (in $x_S$), 
we call $S$ a {\em linear SCC},
otherwise  $S$ is a {\em nonlinear SCC}.
Assume we have already calculated (using R-DNM) an approximation $\tilde{q}_{D(S)}$ to 
the LFP solution $q^*_{D(S)}$ for these lower SCCs.  
We plug in $\tilde{q}_{D(S)}$ into the equations for $S$, obtaining the equation system 
$x_S = P_S(x_S, \tilde{q}_{D(S)})$.
We denote the actual LFP solution of this new equation system by $q'_{S}$. 
(Note that $q'_{S}$ is not necessarily equal to $q^*_{S}$, because $\tilde{q}_{D(S)}$ is
 only an approximation of $q^*_{D(S)}$.)

If $S$ is a nonlinear SCC, we apply a chosen number $g$ of iterations of R-NM on the 
system $x_S = P_S(x_S,\tilde{q}_{D(S)})$
to obtain an approximation $\tilde{q}_S$ of $q'_{S}$;
if $S$ is linear then we just apply 1 iteration of R-NM,
i.e., we solve the linear system 
and round down the solution.  We of course want to make sure our approximations 
are such that $\|q^*_S -  \tilde{q}_S\|_\infty \leq \epsilon$, for all SCCs $S$, and for the
desired additive error $\epsilon > 0$.
We shall establish upper bounds on the number of iterations $g$, and
on the rounding parameter $h$, needed in R-DNM
for this to hold, as a function of various parameters:
the input size $|P|$ and the number $n$ of variables;
the {\em nonlinear depth} $f$ of $P$, which is defined as the maximum,
over all paths of the DAG $H_P$ of SCCs, of the number of nonlinear SCCs on the path;
and the maximum and minimum coordinates of the LFP.

\noindent {\bf Bounds on the size of  LFPs for an MPS.}
For a positive vector $v > 0$, we use $v_{\min} = \min_i v_i$ to denote
its minimum coordinate, and we use $v_{\max} = \max_i v_i$ to denote
its maximum coordinate.   
Slightly overloading notation, for an MPS, $x = P(x)$,  we shall use $c_{\min}$
to denote the minimum value of all positive monomial coefficients and all positive constant
terms in $P(x)$.  Note that $c_{\min}$ also serves as a lower bound for all 
positive constants and coefficients for entries of the Jacobian matrix
$B(x)$, since $B(x)_{ij} = \frac{\partial P_i(x)}{\partial x_j}$.

We prove the following Theorem
in the appendix,
establishing bounds on the maximum and minimum coordinates of the LFP $q^*$
of an MPS $x=P(x)$.

\begin{theorem} \label{thm:q-bounds} 
If $x=P(x)$ is a quadratic MPS in $n$ variables,  with LFP $q^* > 0$,  
and where $P(x)$ has rational coefficients and total encoding size $|P|$ bits, then
\begin{enumerate}

\item  \ \ $q^*_{\min} \geq 2^{-|P|(2^{n}-1)}$, \ \  \mbox{\rm and}\\

\item \ \ 
$q^*_{\max}  \leq 2^{2(n+1)(|P| + 2(n+1)\log(2n+2)) \cdot 5^{n}}$.
\end{enumerate}
\end{theorem}

\noindent {\bf How good are our upper bounds?}  In the
appendix we discuss how good our upper bounds on 
R-DNM are, and in what senses they are optimal, in light of
the convergence rate of Newton's method 
on known bad examples (\cite{lfppoly}), and considerations
relating to the size of $q^*_{\min}$ and $q^*_{\max}$.
In this way, our upper bounds can be seen to be essentially optimal
in several parameters, including the depth of SCCs in
the dependency graph of the MPS, and in terms of $\log \frac{1}{\epsilon}$.

\section{Strongly Connected Monotone Polynomial Systems}

\begin{theorem} \label{scmps} Let $P(x,y)$ be an  $n$-vector of 
monotone polynomials with degree $\leq 2$ in variables which are 
coordinates of the $n$-vector $x$ 
and the $m$-vector $y$, 
where $n \geq 1$ and $m \geq 1$.

Given non-negative $m$-vectors $y_1$ and $y_2$ such that 
$ 0 < y_1 \leq \mathbf{1}$ 
and $0 \leq y_2 \leq y_1$, let $P_1(x) \equiv P(x,y_1)$ 
and $P_2(x) \equiv P(x,y_2)$. 
Suppose that $x=P_1(x)$ is a strongly-connected MPS with LFP solution 
$0 < q^*_1 \leq \mathbf{1}$. \\
Let
 $\alpha = {\min} \{1,c_{\min}\} \mathrm{min } \{y_{\min}, \frac{1}{2}q^*_{\min}\}$, where 
$c_{\min}$ is the smallest non-zero constant or coefficient 
of any monomial in $P(x,y)$, 
where $y_{\min}$ is  
the minimum coordinate of $y_1$,
and finally where $q^*_{\min}$ is the minimum coordinate of $q^*_1$.
Then:

\begin{itemize}
\item[{\bf 1.}] The LFP solution of the MPS 
$x=P_2(x)$ is $q^*_2$ with $0 \leq q^*_2 \leq q^*_1$, and 
\begin{equation*}
\|q^*_1 - q^*_2\|_\infty \leq 
\sqrt{4n \alpha^{-(3n+1)} \|P(\textbf{1}, \textbf{1})\|_\infty \|y_1 - y_2\|_\infty}
\end{equation*} 
Furthermore, if $x=P_1(x)$ is a linear system, then:
$$ \|q^*_1 - q^*_2\|_\infty \leq 2n \alpha^{-(n+2)} \|P(\mathbf{1}, \mathbf{1})\|_\infty \|y_1 - y_2\|_\infty $$

\item[{\bf 2.}] Moreover, for every $0 < \epsilon < 1$,  
if we use $g \geq h-1$ iterations of rounded down
Newton's method
with parameter 
$$h \geq \lceil 
2 +  n \log \frac{1}{\alpha} + \log \frac{1}{\epsilon} \rceil$$
applied to the MPS, $x=P_2(x)$,
starting at $x^{[0]} := \mathbf{0}$,  
to approximate $q^*_2$, 
then the iterations are all defined, and
$\|q^*_2 - x^{[g]} \|_\infty \leq \epsilon$.\end{itemize}
\end{theorem}

Theorem \ref{scmps} and its proof are at the heart of this paper, but
unfortunately the proof is quite involved, and we have no
room to include it.
The proof is in the appendix.
The following easy corollary is also proved in \cite{SEY13}.

\begin{corollary} \label{cor:single-scmps} 
Let  $x=P(x)$ be a strongly connected MPS with $n$ 
variables, and with LFP $q^*$ where $0 < q^* \leq 1$.
Let
 $\alpha = {\min} \{1,c_{\min}\} \frac{1}{2}q^*_{\min}$, where 
$c_{\min}$ is the smallest non-zero constant or coefficient 
of any monomial in $P(x)$. 

Then for all $0 < \epsilon < 1$,  
if we use $g \geq h-1$ iterations of  R-NM 
with parameter 
$h \geq \lceil 2 +  n \log \frac{1}{\alpha} + \log \frac{1}{\epsilon} \rceil$
applied to the MPS, $x=P(x)$,
starting at $x^{[0]} := \mathbf{0}$, 
then the iterations are all defined, and
$\|q^* - x^{[g]} \|_\infty \leq \epsilon$.
\end{corollary}

\section{General Monotone Polynomial Systems}

In this section, we use the rounded-down decomposed Newton's method (R-DNM),
to compute the LFP $q^*$ of general MPSs.  First we consider the case
where $0 < q^* \leq 1$:

\begin{theorem} \label{gen<=1} For all $\epsilon$, where $0 < \epsilon < 1$, 
if $x=P(x)$ is an MPS with LFP solution $0 < q^* \leq 1$, with $q^*_{\min} = \min_i q^*_i$, 
and the minimum non-zero coefficient or constant in $P(x)$ is $c_{\min}$, 
then rounded down decomposed Newton's method (R-DNM)  with parameter 
$$h \geq \left\lceil 3 + 2^f \cdot (\; \log (\frac{1}{\epsilon}) + d \cdot( \log (\alpha^{-(4n+1)}) + \log(16n) + \log (\| P({\mathbf 1}) \|_\infty)) \; ) \right\rceil$$
using $g \geq h-1$ iterations for every nonlinear SCC (and 1 iteration for linear SCC), 
gives an 
approximation $\tilde{q}$ to $q^*$ with $\tilde{q} \leq q^*$ and such that $\|q^* - \tilde{q}\|_\infty \leq \epsilon$.

Here $d$ denotes the maximum depth of SCCs in the DAG $H_P$ of SCCs of 
the MPS $x=P(x)$, 
$f$ is the nonlinear depth,
and $\alpha = {\min} \{1,c_{\min}\} \cdot \frac{1}{2}q^*_{\min}$.
\end{theorem}

Before proving the theorem, let us note that we can obtain worst-case expressions for the needed  number of iterations $g = h-1$,
and the needed rounding parameter $h$, in terms of only $f \leq d \leq n \leq |P|$, and $\epsilon$,  
by noting that $\log(\| P({\mathbf 1})\|_\infty) \leq |P|$, and by
appealing to Theorem \ref{thm:q-bounds} 
to remove references to 
$q^*_{\min}$ in the bounds.
Noting that $c_{\min} \geq 2^{-|P|}$,   
these tell us that $\min \{1,c_{\min}\} \frac{1}{2} q^*_{\min} \geq 2^{-|P|2^n - 1}$. 
Substituting, we obtain that any:

\begin{equation}g \geq \left\lceil 2 + 2^f \cdot (\; \log (\frac{1}{\epsilon}) +  
 d \cdot (|P|2^n(4n+1) + (4n+1) + \log(16n) + |P|) \; ) \right\rceil
\label{eq:bounds-forq<=1-in-terms-of-P-and-eps}
\end{equation}
iterations suffice in the worst case, with rounding parameter $h = g + 1$.
Thus, for $i = \log(1/\epsilon)$ bits of precision,
$g= k_P + c_P \cdot i$ iterations suffice, where $c_P= 2^f$ and
$k_P = O( 2^f 2^n n d |P|)$, with tame constants in the big-O.

\begin{proof}[of Theorem \ref{gen<=1}]
For every SCC $S$, its {\em height} $h_S$ (resp. {\em nonlinear height} $f_S$) 
is the maximum over all paths of the DAG $H_P$ of SCCs starting at $S$, of the number of
SCCs (resp. nonlinear SCCs) on the path. 
We show by induction on the height $h_S$ of each SCC $S$
that $\|q^*_S - \tilde{q}_S\|_\infty \leq \beta^{h_S}  \delta^{2^{-f_S}}$ where $\beta = 
16n \alpha^{-(3n+1)} \|P(\mathbf{1})\|_\infty$ and $\delta = (\frac{\epsilon}{\beta^d})^{2^f}$.
Note that since $n \geq 1$, $\epsilon < 1$, and $\alpha \leq c_{\min}$,
we have $\beta \geq 1$  and $\delta \leq 1$,
and thus also $\delta \leq \sqrt{\delta}$.

Let us first check that this would imply the theorem.   
For all SCCs, $S$, we have $1 \leq h_S \leq d$ and $0 \leq f_S \leq f$, and  thus $\| q^*_S - \tilde{q}_S\|_{\infty} \leq  
\beta^{h_S}  \delta^{2^{-f_S}} \leq \beta^d \delta^{2^{-f}} = \beta^d (\frac{\epsilon}{\beta^d}) = \epsilon$.

We note that $h$ is related to $\delta$ by the following:
\begin{equation} h \geq 2 + n \log \frac{1}{\alpha} + \log \frac{2}{\delta} \label{eq:h-boundb} \end{equation}
This is because $\log \frac{2}{\delta} = 1 +  \log \frac{1}{\delta} = 1 + 2^f(\log \frac{1}{\epsilon} + d\log \beta) =
1+ 2^f(\log (\frac{1}{\epsilon}) + d\log(16n \alpha^{-3n+1} \|P(\mathbf{1})\|_\infty))$.
Note that (3) implies that this inequality holds also for any subsystem
of $x=P(x)$ induced by a SCC $S$ and its successors $D(S)$ because the parameters $n$ and $1/\alpha$ for a subsystem are no larger
than those for the whole system.

We now prove by induction on $h_S$ that
$\|q^*_S - \tilde{q}_S\|_\infty \leq \beta^{h_S}  \delta^{2^{-f_S}}$.

In the base case,  $h_S =1$, we have a strongly connected MPS $x_S = P_S(x)$.
If $S$ is linear, we solve the linear system exactly and then round down to a multiple of $2^{-h}$. Then $f_S = 0$, and we have to show $\|q^*_S - \tilde{q}_S\|_\infty \leq  \beta^{h_S}  \delta^{2^{-f_S}} = \beta \delta$. 
But $\|q^*_S - \tilde{q}_S\|_\infty \leq 2^{-h} \leq \frac{\delta}{2} \leq \beta\delta$.

For the base case where $S$ in non-linear, equation \ref{eq:h-boundb}
and Corollary \ref{cor:single-scmps} imply that
$\|q^*_S - \tilde{q}_S\|_\infty \leq \frac{\delta}{2} $,
which implies the claim since $\delta \leq 1$ and $\beta \geq 1$,
hence $\frac{\delta}{2} \leq \beta^{h_S} \delta^{2^{-f_S}} = \beta^{1} \delta^{2^{-1}}$.

Inductively, 
consider an SCC $S$ with $h_S > 1$.  Then $S$ depends only on SCCs with height at most $h_S-1$. If $S$ is linear, it depends on SCCs of nonlinear depth at most $f_{D(S)} = f_S$, whereas if $S$ is non-linear, it depends on SCCs of nonlinear depth at most $f_{D(S)} = f_S - 1$. We can assume by inductive hypothesis that 
$\|q^*_{D(S)} - \tilde{q}_{D(S)}\|_\infty \leq \beta^{h_S-1}  \delta^{2^{-f_{D(S)}}}$. Take $q'_S$ to be the LFP of $x_S = P_S(x_S, \tilde{q}_{D(S)})$.

Suppose $x_S = P_S(x_S, q^*_{D(S)})$ is linear in $x_S$. 
Then Theorem \ref{scmps} 
with $y_1 := q^*_{D(S)}$ and $y_2 := \tilde{q}_{D(S)}$, yields
 $$\|q^*_S - q'_S\|_\infty \leq 2n_S \alpha^{-(n_S+2)} \|P(\mathbf{1}, \mathbf{1})\|_\infty \|q^*_{D(S)} - \tilde{q}_{D(S)} \|_\infty$$
But $2n_S \alpha^{-(n_S+2)} \|P(\mathbf{1}, \mathbf{1})\|_\infty \leq \frac{\beta}{2}$, so
$\|q^*_S - q'_S\|_\infty \leq   \frac{\beta}{2} \|q^*_{D(S)} - \tilde{q}_{D(S)} \|_\infty 
 \leq  \frac{\beta}{2}  \beta^{h_S-1} \delta^{2^{-f_S}}=  \frac{1}{2}\beta^{h_S}  \delta^{2^{-f_S}}$. 
Since $\|q'_S-\tilde{q}_S\|_\infty \leq 2^{-h} \leq \frac{\delta}{2} \leq \frac{1}{2}\beta^{h_S}  \delta^{2^{-f_S}}$, it follows that 
$\|q^*_S - \tilde{q}_S\|_\infty \leq \beta^{h_S}  \delta^{2^{-f_S}}$.

Suppose that  $x_S = P_S(x_S, q^*_{D(S)})$ is non-linear in $x_S$. 
Theorem \ref{scmps}, with $y_1 := q^*_{D(S)}$ and $y_2 := \tilde{q}_{D(S)}$, yields that 
\begin{equation}
\|q^*_S - q'_S\|_\infty \leq \sqrt{4n \alpha^{-(3n+1)} \|P(\mathbf{1})\|_\infty \|q^*_{D(S)} - (\tilde{q})_{D(S)}\|_\infty}
\label{eq:bound-induct-decomposed-ineqb}
\end{equation}
Note that the $\alpha$ from Theorem \ref{scmps} is indeed the same or better (i.e., bigger) than the $\alpha$ in this Theorem, because 
$y_{\min} = (q^*_{D(S)})_{\min} \geq q^*_{\min}$ and $(q^*_S)_{\min} \geq q^*_{\min}$. 
Rewriting (\ref{eq:bound-induct-decomposed-ineqb}) 
in terms of $\beta$, we have  $\|q^*_S - q'_S\|_\infty \leq \sqrt{ \frac{1}{4}\beta \|q^*_{D(S)} - (\tilde{q})_{D(S)}\|_\infty}$.
By inductive assumption, 
$\|q^*_{D(S)} - \tilde{q}_{D(S)}\|_\infty \leq \beta^{h_S-1}  \delta^{2^{-f_S+1}}$, and thus
$\|q^*_S - q'_S\|_\infty \leq \sqrt{ \frac{1}{4} \beta^{h_S}  \delta^{2^{1-f_S}} } \leq \frac{1}{2} \beta^{h_S} \delta^{2^{-f_S}}$. 
Thus to
 show that the inductive hypothesis holds also for SCC $S$, it suffices 
to show that for the approximation $\tilde{q}_S$ we have $\|q'_S - \tilde{q}_S\|_\infty \leq \frac{1}{2} \beta^{h_S} \delta^{2^{-f_S}}$. 
But $\beta \geq 1$, $h_S \geq 1$,  $2^{-f_S} \leq 1$ and $\delta \leq 1$, so 
$\frac{1}{2}  \delta \leq \frac{1}{2} \beta^{h_S} \delta^{2^{-f_S}}$, so
it suffices to show that 
$\|q'_S - \tilde{q}_S\|_\infty \leq \frac{1}{2}  \delta$.
Part 2 of Theorem \ref{scmps} tells us that we will have
$\|q'_S - \tilde{q}_S\|_\infty \leq \frac{1}{2}  \delta$
if $g \geq h-1$ and $h \geq 2 + n \log \frac{1}{\alpha} + \log \frac{2}{\delta}$. 
But we have already established this in equation (\ref{eq:h-boundb}),
hence the claim follows.

\qed   \end{proof}

Next, we want to generalize Theorem \ref{gen<=1} to arbitrary MPSs that have an LFP, $q^* > 0$,  without the restriction
that $0 < q^* \leq 1$.
The next Lemma allows us to establish this by a suitable ``rescaling'' of any MPS which has an LFP $q^* > 0$.
If $x=P(x)$ is a MPS and $c >0$, we can consider the MPS $x=\frac{1}{c}P(cx)$. 

\begin{lemma}  \label{scale} Let $x=P(x)$ be a  MPS with LFP solution $q^*$, 
and with Jacobian $B(x)$, and recall that for $z \geq 0$, 
${\mathcal N_P}(z) := 
z + (I-B(z))^{-1}(P(z)-z)$
denotes the Newton operator applied at $z$ on $x=P(x)$.
Then:
\begin{itemize}
\item[(i)] The LFP solution of $x=\frac{1}{c}P(cx)$ is $\frac{1}{c}q^*$.
\item[(ii)] The Jacobian of $\frac{1}{c}P(cx)$ is $B(cx)$.
\item[(iii)] A Newton iteration of the ``rescaled'' MPS, $x=\frac{1}{c}P(cx)$, 
applied to the vector $z$ is given by $\frac{1}{c}{\mathcal N_P}(cz)$.
\end{itemize}
\end{lemma}
\begin{proof} From \cite{rmc}, we know that the value iteration sequence  
$P(0)$, $P(P(0))$, $P(P(P(0))) \ldots P^k(0)$ converges to $q^*$. 
Now note that
for the MPS $x=\frac{1}{c}P(cx)$, the value iteration sequence 
is $\frac{1}{c}P(0)$, $\frac{1}{c}P(c \frac{1}{c}P(0)) = \frac{1}{c}P(P(0))$,  $\frac{1}{c}P(P(P(0)))$... which thus converges to 
$\frac{1}{c}q^*$. This establishes (i).

For (ii), note that, by the chain rule in multivariate
calculus (see, e.g., \cite{apostol74} Section 12.10), the Jacobian of $P(cx)$ is $cB(cx)$.
Now (iii) follows because:
\begin{equation*}
z + (I-B(cz))^{-1}( \frac{1}{c}P(cz) - z) = \frac{1}{c} (cz + (I-B(cz))^{-1}(P(cz)-cz)) = \frac{1}{c}{\mathcal N_P}(cz).
\label{eq:rescaled-newton-iterate}
\end{equation*}
\qed   \end{proof}

We use Lemma \ref{scale}  to generalise Theorem \ref{gen<=1} to MPSs with LFP $q^*$, where $q^*$ does
not satisfy $q^* \leq 1$.

\begin{theorem} \label{gen}
If $x=Q(x)$ is an MPS with $n$ variables, with LFP solution $q^* > 0$,  
if $c'_{\min}$ is the least positive coefficient of any monomial in $Q(x)$, 
then R-DNM 
with rounding parameter $h'$, and using $g'$ iterations per nonlinear SCC
(and one for linear), gives an approximation 
$\tilde{q}$ such that $\|q^* - \tilde{q}\|_\infty \leq \epsilon'$, where
$$g'= 2 + \lceil \; 2^f \cdot (\log (\frac{1}{\epsilon'}) + d \cdot (2u   +  \log (\alpha'^{-(4n+1)}) + \log(16n) + \log (\| Q({\mathbf 1}) \|_\infty)) \; ) \ \rceil $$ 
and 
$h' = g'+1 - u$, 
where $u = \max \{ 0, \lceil \log q^*_{\max} \rceil \}$,  $d$ is the maximum depth of SCCs in the DAG $H_Q$ of SCCs of 
$x=Q(x)$, $f$ is the nonlinear depth, and $\alpha' = 2^{-2u} \min \{1,c'_{\min}\} \min \{1, \frac{1}{2}q^*_{\min}\}$.
\end{theorem}

We can again obtain worst-case expressions for the needed  
number of iterations $g'$,
and the needed rounding parameter $h'$, in terms of only $f \leq d \leq n \leq |Q|$, and $\epsilon'$,  
by noting that $\log(\| Q({\mathbf 1}) \|_\infty) \leq |Q|$ and by
appealing to Theorem \ref{thm:q-bounds} 
to remove references to 
$q^*_{\min}$ and $q^*_{\max}$ in the bounds.  
Substituting and simplifying we get that to guarantee 
additive error at most $\epsilon'$, i.e. for $i = \log(1/\epsilon')$
bits of precision, it suffices in the worst-case to
apply $g'=k_Q + c_Q \cdot i$ iterations of R-DNM with rounding
parameter $h'=g'+1$ 
(which is more accurate rounding than $h'=g'+1-u$),
where $c_Q = 2^f$, and
$k_Q = O(2^f 5^n n^2 d (|Q|+n \log n))$
(and we can calculate precise, tame, constants for the big-O expression).

\begin{corollary}  If $x=P(x)$ is an MPS with LFP solution $q^*$  with $0 < q^*_{\min} \leq q^*_i \leq q^*_{\max} $ for all $i$, with the least coefficient of any monomial in $P(x)$, $c_{\min}$, with  $f$  the nonlinear depth of the DAG of SCCs of $x=P(x)$ and with encoding size $|P|$ bits, we can compute an approximation $\tilde{q}$ to $q^*$ with $\|q^* - \tilde{q}\|_\infty \leq \epsilon$, for any given $0 < \epsilon \leq 1$, in 
time polynomial in $|P|$,$2^f$, $\log \frac{1}{\epsilon}$,$\log \frac{1}{q^*_{\min}}$ and $\log q^*_{\max}$.\end{corollary}
\begin{proof} After preprocessing to remove all variables $x_i$ with $q^*_i =0$, which takes P-time in $|P|$,
we use R-DNM as specified in Theorem \ref{gen}.
Calculating a Newton iterate at $z$ is just a matter of solving a matrix equation and if the coordinates of $z$ are multiples of $2^{-h}$ 
this can be done in time polynomial in $|P|$ and $h$. Theorem \ref{gen} 
tells us that the number of iterations and $h$ are polynomial in 
$2^f$, $\log \frac{1}{\epsilon}$, $\log \frac{1}{q^*_{\min}}$, $\log q^*_{\max}$, $n$,
$\log \frac{1}{c_{\min}}$ and $\log \|P(\textbf{1})\|_\infty$. The last three of these are bounded by $|P|$. 
Together, these give the corollary. \qed   \end{proof}

\section{MPSs and Probabilistic 1-Counter Automata}

A {\bf probabilistic 1-counter automaton} (p1CA), $M$, is a $3$-tuple 
$M= (V,
\delta, \delta_0)$ where $V$ is a finite set of {\em control states} and
$\delta \subseteq V \times \mathbb{R}_{>0} \times \set{-1, 0, 1}
\times V$ and $\delta_0 \subseteq V \times \mathbb{R}_{>0} \times
\set{0, 1} \times V$ are {\em transition relations}.
The transition relation $\delta$ is enabled when the counter is nonzero,
and the transition relation $\delta_0$ is enabled when it is zero.
For example, a transition of the form,  $(u,p,-1,v) \in \delta$,
says that if the counter value is positive, and we are currently
in control state $u$, then with probability $p$ we move in
the next step to control state $v$ and we decrement the counter by 1.
A p1CA defines in the obvious way an underlying countably infinite-state 
(labeled) Markov chain, whose set of configurations (states) 
are pairs $(v,n) \in V \times \nat$.   
A {\em run} (or {\em trajectory}, 
or {\em sample path}), starting at initial state $(v_0,n_0)$ 
is defined in the usual way, as 
a sequence of configurations  $(v_0,n_0), (v_1,n_1), (v_2,n_2), \ldots$
that is consistent with the transition relations of $M$.

As explained in \cite{EWY-qbd-2010},
p1CAs are in a precise sense equivalent to 
discrete-time {\em quasi-birth-death processes} (QBDs),
and to {\em 1-box recursive Markov chains}.

Quantities that play a central role for the analysis of QBDs
and p1CAs  (both for transient analyses and steady-state analyses,
as well as for model checking)
are their {\em termination probabilities}
(also known as their {\em $G$-matrix} in the QBD literature, 
 see, e.g., \cite{LatRam99,BiLaMe05,EWY-qbd-2010}).
These are defined as the probabilities, $q^*_{u,v}$,
of hitting counter value $0$ for
the first time in control state $v \in V$,  
when starting in configuration $(u,1)$.

Corresponding to the termination probabilities
of every QBD or p1CA is  
a special kind of MPS, $x=P(x)$,
whose LFP solution $q^*$ gives the termination probabilities of the p1CA.
The MPSs corresponding to p1CAs have the following special structure.
For each pair of control states $u,v \in V$ of the p1CA, there is 
a variable $x_{uv}$.
The equation for each variable $x_{uv}$ has the following form:
\begin{equation}\label{eq:xuv}
x_{uv} =
 p^{(-1)}_{uv} + \bigg(\sum_{w \in V} p^{(0)}_{uw} x_{wv}\bigg) +
\sum_{y \in V} p^{(1)}_{uy}
\sum_{z \in V} x_{yz} x_{zv}
\end{equation}
where
for all states $u,v \in V$, and $j \in \{-1,0,1\}$, 
the coefficients $p^{(j)}_{uv}$ are non-negative transition probabilities
of the p1CA,
and such that for all states $u \in V$, $\sum_{j \in \{-1,0,1\}} \sum_{v \in V} p^{(j)}_{uv} \leq 1$.
We can of course clean up this MPS in P-time (by Proposition \ref{prop:zero-removal-in-p-time}), to remove all variables
$x_{uv}$ for which $q^*_{u,v} = 0$.   In what follows,
we assume this has been done, and thus that for the remaining
variables $0 < q^* \leq 1$.

In \cite{EWY-qbd-2010}, the decomposed Newton's method (DNM)
is used {\em with exact arithmetic}
in order to approximate the  LFP for p1CAs using polynomially many arithmetic
operations, i.e., in polynomial time in the {\em unit-cost arithmetic model of computation}. 
However \cite{EWY-qbd-2010} did not establish any result about the rounded down 
version of DNM, and thus no results
on the time required in the standard Turing model of computation.
We establish instead results about 
R-DNM applied to the MPSs arising from p1CAs, in order to turn this 
method into a P-time
algorithm in the standard model of computation.

It was shown in \cite{EWY-qbd-2010} that in any path through the DAG of SCCs of
the dependency graph for the MPS 
associated with a p1CA, $M$, there is at most one non-linear SCC,
i.e. the nonlinear depth is $\leq 1$.
Also, \cite{EWY-qbd-2010} obtained a lower bound
on $q^*_{\min}$, the smallest positive  termination probability.  
Namely, if $c_{\min}$ denotes the smallest positive
transition probability of a p1CA, $M$, and thus also the
smallest positive constant or coefficient of any monomial
in the corresponding MPS, $x=P(x)$, they showed:

\begin{lemma}{(Corollary 6 from \cite{EWY-qbd-2010})} 
$q^*_{\min} \geq c_{\min}^{r^3}$, where $r$ is the number of control states of the p1CA.
\label{lem:low-bound-probs-in-qbds}
\end{lemma}

They used these results to bound the {\em condition number} of the Jacobian
matrix for each of the linear SCCs, and to thereby show that one
can approximate $q^*$ in polynomially many arithmetic operations using
decomposed Newton's method. Here, we get a stronger result,
placing the problem of computing termination probabilities
for p1CA in P-time in the standard Turing model,
using the results from this paper:

\begin{theorem}  Let $x=P(x)$ be the MPS associated with p1CA, $M$, 
let $r$ denote the number of
control states of $M$, and let $m$ denote the maximum number
of bits required to represent the numerator and denominator of
any positive rational transition probability in $M$.

Apply R-DNM, including rounding down linear SCCs, to the MPS $x=P(x)$,
using rounding parameter
$h := 8mr^7 + 
2mr^5 + 9r^2 + 3 + \lceil 2 \log \frac{1}{\epsilon} \rceil$
and such that for each non-linear SCC we perform $g = h-1$ iterations,
whereas for each linear SCC we only perform $1$ R-NM iteration.  

This algorithm computes an approximation $\tilde{q}$ to $q^*$, such that $\|q^* - \tilde{q}\|_\infty < 
\epsilon$.
The algorithm runs in time polynomial in $| M|$ and 
$\log \frac{1}{\epsilon}$, in the standard Turing model of computation.
\label{thm:qbd-termin-p-time}\end{theorem}

This follows from
Theorem \ref{gen<=1}, using the fact that
$\log(1/q^*_{\min})$ is polynomially bounded by
Lemma \ref{lem:low-bound-probs-in-qbds}, and 
the fact that the nonlinear depth of the MPS $x=P(x)$ for any p1CA is $f \leq 1$
(\cite{EWY-qbd-2010}).
The detailed proof is in the appendix.

\subsection{Application to $\omega$-regular model checking for p1CAs}

Since computing termination probabilities of p1CAs 
(equivalently, the $G$-matrix
of QBDs) 
plays such a central role in other analyses 
(see, e.g., \cite{LatRam99,BiLaMe05,EWY-qbd-2010,BKK11}), 
the P-time algorithm given in the previous section for computing termination
probabilities of a p1CA  (within arbitrary desired precision)
directly facilitates 
P-time algorithms for various other important problems.

Here we highlight just one of these applications:  a P-time 
algorithm {\em in the Turing model of computation} 
for model checking a p1CA with respect to any
$\omega$-regular property.  An analogous result 
was established by Brazdil, Kiefer, and Kucera \cite{BKK11}
in the unit-cost RAM model of computation.

\begin{theorem}
\label{thm:qbd-mcing}
Given a p1CA, $M$, with states labeled from 
an alphabet $\Sigma$, and
with a specified initial control state $v$, and given an $\omega$-regular
property $L(B) \subseteq \Sigma^\omega$, 
which is specified by a non-deterministic 
B\"{u}chi automaton, $B$,  let $Pr_{M}(L(B))$ denote the probability that a run of
$M$ starting at configuration $(v,0)$ generates an $\omega$-word
in $L(B)$.  There is an algorithm that, for any $\epsilon > 0$,
computes an additive 
$\epsilon$-approximation, $\tilde{p} \geq 0$, of $Pr_{M}(L(B))$,
i.e., with $| Pr_{M}(L(B)) - \tilde{p} | \leq \epsilon$.
The algorithm runs in time polynomial in $|M|$,
$\log \frac{1}{\epsilon}$, and $2^{|B|}$,   
in the standard Turing model of computation.
\end{theorem}

\newpage

\appendix

\section{Background Lemmas, and  Missing Proofs}

\noindent We first recall some Lemmas from \cite{ESY12}:\footnote{
In \cite{ESY12}, the statements of Lemma 3.3 and 3.4 assume that the MPS
is in SNF form, but as noted 
in \cite{ESY12}, the proofs of Lemma 3.3 and 3.4 do not
require that $x=P(x)$ is in SNF form, nor that it is an MPS,
only that it is quadratic.}

\begin{lemma}{(Lemma 3.3 of \cite{ESY12})}
 Let $x=P(x)$ be a quadratic MPS, with $n$ variables,
and let $a, b \in \real^n$. Then:
$$P(a)-P(b) = B(\frac{a + b}{2})(a-b)= \frac{B(a) + B(b)}{2}(a-b)$$
\label{lem:restate-lem-3-3}
\end{lemma}

\begin{lemma} \label{newton-from-esy12} 
Let $x=P(x)$ be a quadratic MPS.
Let $z \in \real^n$ be any vector  
such that $(I - B(z))$ is non-singular, and thus 
$\mathcal{N}_P(z)$ is defined. 
Then:
$$q^* -  \mathcal{N}_P(z) = (I-B(z))^{-1}\frac{B(q^*) - B(z)}{2}(q^* - z)$$
\end{lemma}

\noindent We will also need the following lemma from \cite{lfppoly}:

 \begin{lemma}{(Lemma 5.4 from \cite{lfppoly}, Lemma 3.7 from \cite{ESY12})} 
\label{cone} 
Let $x=P(x)$ be a MPS, with polynomials of 
degree bounded by 2, with LFP, $q^* \geq 0$.
Let $B(x)$ denote the Jacobian matrix of $P(x)$.
For any positive vector 
$\textbf{d} \in \mathbb{R}^n_{> 0}$ 
that satisfies $B(q^*) \textbf{d} \leq \textbf{d}$, any positive real value $\lambda > 0$, 
and any nonnegative vector $z \in \real^n_{\geq 0}$, 
if $q^* - z \leq \lambda \textbf{d}$, and $(I-B(z))^{-1}$
exists and is nonnegative, then 
$$q^* - {\mathcal N}_P(z) \leq \frac{\lambda}{2} \textbf{d}$$\end{lemma}

We also need to recall a number of basic facts from matrix analysis and Perron-Frobenius theory.
\noindent For a square matrix $A$, let $\rho(A)$ denote the spectral radius of $A$. Recall that 
a nonnegative square matrix $A$ is called
{\em irreducible} if its underlying directed graph
is strongly connected, where the adjacency matrix
of its underlying directed graph is obtained by setting the positive entries of the matrix $A$
to $1$.

\begin{lemma}{(see, e.g., \cite{HornJohnson85}, Theorem 8.4.4)}
If $A$ is an irreducible nonnegative square matrix, then 
there is a positive eigenvector $v > 0$, such that $Av = \rho(A)v$.
Such a vector $v$ is called the {\em Perron vector} of $A$.
It is unique up to rescaling by a positive factor.
\label{lem:perron}
\end{lemma}

\begin{lemma}{(see, e.g., \cite{LanTis85}, Theorem 15.4.1 and Exercise 1, page 540)} If $A$  is an irreducible nonnegative square matrix and 
$0 \leq B \leq A$, but $B \not= A$, then 
$\rho(B) < \rho(A)$.
\label{lem:irr-non-neg-spec-lower}\end{lemma}
\begin{lemma}{(see, e.g., \cite{LanTis85}, Theorem 15.2.2, page 531)} If $A$ is a square matrix with $\rho(A) < 1$, then $I-A$ is non-singular and 
$(I-A)^{-1} = \sum_{i=0}^\infty A^i$.
\label{lem:inverse-of-sub-1-spec-exists}\end{lemma}

\begin{lemma}{(see, e.g., \cite{LanTis85}, Section 15.3 and Exercise 11)} If $A$  is an irreducible nonnegative square matrix,
and $v > 0$ is a positive eigenvector associated with some eigenvalue  $r$, i.e., such that $A v = r v$,
then $r = \rho(A)$.
Thus $v > 0$ is the Perron vector (which is unique up to scaling).
\label{lem:irr-unique-positive-eigenvector}
\end{lemma}

\subsection{Proof of Theorem \ref{scmps}}

\label{sec-app:proof-of-scmps}

\noindent 
{\bf Theorem  \ref{scmps}.} 
{\em Let $P(x,y)$ be an  $n$-vector of 
monotone polynomials with degree $\leq 2$ in variables which are 
coordinates of the $n$-vector $x$ 
and the $m$-vector $y$, 
where $n \geq 1$ and $m \geq 1$.

Given non-negative $m$-vectors $y_1$ and $y_2$ such that 
$ 0 < y_1 \leq \mathbf{1}$ 
and $0 \leq y_2 \leq y_1$, let $P_1(x) \equiv P(x,y_1)$ 
and $P_2(x) \equiv P(x,y_2)$. 
Suppose that $x=P_1(x)$ is a strongly-connected MPS with LFP solution 
$0 < q^*_1 \leq \mathbf{1}$. \\
Let
 $\alpha = {\min} \{1,c_{\min}\} \mathrm{min } \{y_{\min}, \frac{1}{2}q^*_{\min}\}$, where 
$c_{\min}$ is the smallest non-zero constant or coefficient 
of any monomial in $P(x,y)$, 
where $y_{\min}$ is  
the minimum coordinate of $y_1$,
and finally where $q^*_{\min}$ is the minimum coordinate of $q^*_1$.
Then:

\begin{itemize}
\item[{\bf 1.}] The LFP solution of the MPS 
$x=P_2(x)$ is $q^*_2$ with $0 \leq q^*_2 \leq q^*_1$, and 
\begin{equation}
\|q^*_1 - q^*_2\|_\infty \leq 
\sqrt{4n \alpha^{-(3n+1)} \|P(\textbf{1}, \textbf{1})\|_\infty \|y_1 - y_2\|_\infty}
\label{eq:bound-on-diff-of-lfps}
\end{equation} 

Furthermore, if $x=P_1(x)$ is a linear system, then:
\begin{equation} 
\|q^*_1 - q^*_2\|_\infty \leq 2n \alpha^{-(n+2)} \|P(\mathbf{1}, \mathbf{1})\|_\infty \|y_1 - y_2\|_\infty 
\label{eq:bound-on-lin-diff-of-lfps}
\end{equation}

\item[{\bf 2.}] Moreover, for every $0 < \epsilon < 1$,  
if we use $g \geq h-1$ iterations of rounded down
Newton's method
with parameter 
$$h \geq \lceil 2 +  n \log \frac{1}{\alpha} + \log \frac{1}{\epsilon} \rceil$$
applied to the MPS, $x=P_2(x)$,
starting at $x^{[0]} := \mathbf{0}$,  
to approximate $q^*_2$, 
then the iterations are all defined, and
$\|q^*_2 - x^{[g]} \|_\infty \leq \epsilon$.\end{itemize}
}

\begin{proof}
We first establish {\bf 1.}
Since $x=P_1(x)$ is
a strongly connected system of equations, and $q^*_1 >  0$, this implies that matrix $B_1(q^*_1)$ is 
non-negative and irreducible,
where $B_1(x)$ is the Jacobian matrix of $P_1(x)$.   

Thus, by Lemma \ref{lem:perron},
there is a positive Perron eigenvector $v >0$ of $B_1(q^*_1)$, 
which satisfies $B_1(q^*_1)v = \rho(B_1(q^*_1)) v$.
We can 
always scale $v$ such that 
 $\|v\|_\infty = 1$.

We will observe that $B_1(q^*_1)v \leq v$, and that if we scale $v$ so that $\|v\|_\infty = 1$ 
then the smallest coordinate 
of $v$, denoted $v_{\min}$, has $v_{\min} \geq \alpha^n$.

\begin{lemma}(This is a variant of  Lemma 6.5 from \cite{rmc})
For any strongly-connected MPS,  $x= P(x)$,  with LFP $q^* > 0$, 
and Jacobian $B(x)$,  
we have $\rho(B(q^*)) \leq 1$, and for all vectors $y$ 
with $0 \leq y < q^*$,  $\rho(B(y)) < 1$.
\label{lem:spectr-at-lfp-is-at-most-1}
\end{lemma} 
\begin{proof}
We will only
show here that $\rho(B(q^*)) \leq 1$ if $x=P(x)$ is strongly
connected, but in fact this holds 
 for any MPS, $x=P(x)$, with LFP $q^* > 0$.
We do so because we will only use the strongly-connected case.

If we have $0 \leq z \leq y$ and  $z \leq P(z)$,
then Lemma 6.4 of \cite{rmc} shows that for any $d \geq 1$, $B^d(z)(y-z) \leq P^d(y) - P^d(z)$.
Let $x^i = P^i(\mathbf{0})$, for all $i \geq 1$.  
Recall that $\lim_{i \rightarrow \infty} x^i = q^*$. 
Also note that,
because $x=P(x)$ is strongly connected,  $x^i < q^*$ for
all $i$.  

Then for all $i, d \geq 1$, $B^d(x^i)(q^* - x^i) \leq P^d(q^*) - P^d(x^i) = q^*- x^{i+d}$.
But since $\lim_{d \rightarrow \infty} x^{i+d} = q^*$, we see that the right hand side goes to $0$.
But since $(q^*- x^i) > 0$ for all $i$,  it must be the case that
$B^d(x^i) \rightarrow 0$, as $d$ goes to infinity.
But this is a necessary and sufficient condition for $\rho(B(x^i)) < 1$.
Now notice that for any vector $y$ such that $0 \leq y < q^*$, there is some 
$i$ such
that $y \leq x^i$.  Thus, by monotonicity of $\rho(B(x))$ in $x \geq 0$, 
we must have $\rho(B(y)) < 1$.

Thus, also, since $\lim_{i \rightarrow \infty} x^i = q^*$, and by continuity of the spectral radius
function, we get that $\rho(B(q^*)) \leq 1$. 
\qed   \end{proof}
 
\begin{corollary}
$\rho(B_1(q^*_1)) \leq 1$,  and thus if $v$ is a Perron
vector of $B_1(q^*_1)$ then
$B_1(q^*_1)v = \rho(B_1(q^*_1))v \leq v$.  
\label{cor:perron-cone}
\end{corollary}

\noindent The following basic lemma, applied to 
$B_1(q^*_1)$ and its normalized Perron vector $v$, yields the desired result about $v$:
\begin{lemma} \label{coneratio} If $A$ is a irreducible, non-negative $n \times n$ matrix with minimum non-zero entry 
$a_{\min}$, and $u \geq 0$ is a non-zero vector in $\mathbb{R}^n$ with $Au \leq u$, then 
$a_{\min} \leq 1$ and if the minimum and maximum coordinates of $u$ are
denoted $u_{\min}$ and $u_{\max}$, respectively, then we have 
$\frac{u_{\min}}{u_{\max}} \geq a_{{\min}}^n$. In particular $u > 0$.\end{lemma}
\begin{proof}
Let $i$,$j$ be some coordinates with $u_i = u_{\min}$ and $u_j = u_{\max}$. Because $A$ is irreducible and non-negative, there is a power $0 \leq k \leq n$ with $(A^k)_{ij} > 0$.  
By matrix multiplication, for any $k \geq 1$,  $(A^k)_{ij} = \sum \prod_l A_{i_l,i_{l+1}}$, where the sum is taken 
over all length $k+1$ sequences of indices $i_1, \ldots, i_{k+1}$, with 
$i_1 =i$ and $i_{k+1} = j$, and  with $i_l \in \{1,\ldots,n\}$ for all $l$ ranging from $1$ to $k$. 
At least one of these products is non-zero and thus it is  at least $a_{\min}^k$. 
That is $(A^k)_{ij} \geq a_{\min}^k$.  Since $A u \leq u$,
and $A$ is non-negative, 
a simple induction gives that $A^ku \leq u$. 
And since $u$ is non-zero,  $u_{\max} = u_j > 0$, so $0 < A^k_{ij} u_j \leq u_i$. 
Since $u_i = u_{\min}$, this means $u > 0$.  Also,
$1 \geq \frac{u_{\min}}{u_{\max}} = \frac{u_i}{u_j} \geq A^k_{ij} \geq a_{\min}^k$.
Note that since $1 \geq a_{\min}^k$, this implies $a_{\min} \leq 1$.
We know that $1 \leq k \leq n$, so $a_{\min}^k \geq a_{\min}^n$. \qed   \end{proof}

Applying Lemma \ref{coneratio} to $A= B_1(q^*_1)$ and $v$ the Perron vector of $B_1(q^*_1)$,
normalized so that $v_{\max} = 1$, 
and observing that the smallest non-zero entry of $B_1(q^*_1)$ is at least $\alpha$, we get:

\begin{corollary}
If $v$ is the Perron 
vector of $B_1(q^*_1)$, normalized so that $v_{\max} = 1$,
then 
$\frac{v_{\min}}{v_{\max}} = v_{\min} \geq \alpha^n$.
\label{cor:perron-low-bound-ratio}
\end{corollary}

Next, to show that $0 \leq q^*_2 \leq q^*_1$, we consider $P_1^k(0) =
P_1(P_1(...P_1(0)...))$, i.e., the $k$'th iterate of $P_1$ applied to
the vector ${\mathbf 0}$, and $P_2^k(0)$.  
We know 
that for any MPS, $x=P(x)$ with LFP $q^* \in \real^n_{\geq 0}$, we
have $\lim_{k \rightarrow \infty} P^k(0) = q^*$   (\cite{rmc}).
Thanks to the monotonicity of $P$, for any $x \geq 0$, we have 
$P_1(x) \geq P_2(x)$. By the monotonicity of $P_1$ and an easy induction, $P_1^k(0) \geq P_2^k(0)$. 
So $q^*_1 \geq q^*_2$.

Next we want to obtain the bounds 
(\ref{eq:bound-on-diff-of-lfps}) 
and (\ref{eq:bound-on-lin-diff-of-lfps})
on $\|q^*_1 - q^*_2\|_\infty$.
If $q^*_1 = q^*_2$, then we are trivially done so we assume 
that $q^*_2 \neq q^*_1$. Because $x=P_1(x)$ is 
at most quadratic, we can apply Lemma \ref{lem:restate-lem-3-3} to get:
\begin{equation}
B_1(\frac{1}{2}(q^*_1+q^*_2)) (q^*_1 - q^*_2) = P_1(q^*_1) - P_1(q^*_2)  = q^*_1 - P_1(q^*_2)
\label{eq:from-lem-3-3-esy12}
\end{equation}
Multiplying both sides of equation  (\ref{eq:from-lem-3-3-esy12})
by $-1$, and then adding $(q^*_1 - q^*_2)$ to both sides, 
we get:

\begin{eqnarray}
(I - B_1(\frac{1}{2}(q^*_1+q^*_2))) (q^*_1 - q^*_2) & = & 
(q^*_1 - q^*_2) - (q^*_1 - P_1(q^*_2))  \nonumber \\ 
& = & 
P_1(q^*_2) - q^*_2  \nonumber \\ & = & 
{P_1(q^*_2) - P_2(q^*_2) 
\label{eq:mid-deriv-diff-bound}}
\end{eqnarray}

Provided that $(I - B_1(\frac{1}{2}(q^*_1+q^*_2))$ is non-singular,
we can multiply both sides of equation 
(\ref{eq:mid-deriv-diff-bound})
by $(I - B_1(\frac{1}{2}(q^*_1+q^*_2))^{-1}$,
to get
\begin{equation}
q^*_1 - q^*_2  =  (I - B_1(\frac{1}{2}(q^*_1+q^*_2))^{-1}(P_1(q^*_2) - P_2(q^*_2))
\label{eq:iden-deriv-diff-bound}
\end{equation}

We will be taking the $\|.\|_\infty$ norm of equation 
(\ref{eq:iden-deriv-diff-bound})
to obtain the bound we need for $\|q^*_2 - q^*_1\|_\infty$. 
To do this we first need to bound 
$\|(I - B_1(\frac{1}{2}(q^*_1+q^*_2))^{-1}\|_\infty$, and in particular 
we need to show that
$(I - B_1(\frac{1}{2}(q^*_1+q^*_2))^{-1}$ is nonsingular.

By (\ref{eq:from-lem-3-3-esy12}) we have $q^*_1 - P_1(q^*_2) = 
B_1(\frac{1}{2}(q^*_1+q^*_2)) (q^*_1 - q^*_2)$. Now 
$P_1(q^*_2) \geq P_2(q^*_2)=q^*_2$. 
Thus $q^*_1 - q^*_2 \geq q^*_1 - P_1(q^*_2)$.
So $B_1(\frac{1}{2}(q^*_1+q^*_2)) (q^*_1 - q^*_2) \leq (q^*_1 - q^*_2)$. 
Since each polynomial in $P(x,y)$ has degree no more than 2, each entry of 
$B_1(x)$ is a polynomial of degree no more than 1 in both $x$ and in the entries
of $y_1$ when these are treated as variables.
In other words, each entry of $B_1(x)$ can be expressed in
the form $(\sum_i c_i x_i) + (\sum_j c'_j y_j) + c''$, where $c_i, c'_j,$ and
 $c''$ are all non-negative 
coefficients and constants of $P(x,y)$ (possibly multiplied
by $2$ in the case where the term of $P(x,y)$ they originate from is of the form $c x_r^2$) 
for all indices $i$ and $j$. 
So for any $i$,$j$ $B_1(\frac{1}{2}q^*_1)_{ij} \geq \frac{1}{2} B_1(q^*_1)_{ij}$. Also, 
since $q^*_2 \leq q^*_1$ and $q^*_1 > 0$, we have
$B_1(\frac{1}{2}(q^*_1+q^*_2)) \geq B_1(\frac{1}{2}q^*_1)$, and the
matrices  $B_1(\frac{1}{2}(q^*_1+q^*_2))$ and   $B_1(\frac{1}{2}q^*_1)$ are both 
irreducible.  Also, both these matrices  
have non-zero entries $\geq \alpha$,
because the coefficients $c_i, c'_j$, and $c''$ are 
all $\geq c_{\min}$, and the entries of $\frac{1}{2}q^*_1$ and $\frac{1}{2}(q^*_1 + q^*_2)$ are both $\geq \frac{1}{2}{q^*_{\min}}$. 
Now, Lemma \ref{coneratio},
applied to matrix $A = B_1(\frac{1}{2}(q^*_1+q^*_2))$ and
vector $u = (q^*_1 - q^*_2)$,  yields that 
\begin{equation}
\frac{(q^*_1 - q^*_2)_{\min}}{(q^*_1 - q^*_2)_{\max}} \geq \alpha^{n}
\label{eq:ratio-low-q1-q2}
\end{equation}
In particular, we have thus also shown that if $q^*_2 \neq q^*_1$ then: 
\begin{equation}
q^*_2 < q^*_1
\label{eq:q2-less-q1}
\end{equation}

Now suppose that $B_1(x)$ is not independent of $x$.
Since $q^*_2 < q^*_1$, there is some entry of $B_1(\frac{1}{2}(q^*_1+q^*_2))$,
say $B_1(\frac{1}{2}(q^*_1+q^*_2))_{ij}$, which 
is strictly smaller than that of $B_1(q^*_1)_{ij}$.
The entry $B_1(x)_{ij}$ must be of the form 
 $(\sum_i c_i x_i) + (\sum_j c'_j y_j) + c''$, where 
for some $k$,  $c_k > 0$ so that the term $c_k x_k$ depends on $x_k$.
We must therefore have $B_1(q^*_1)_{ij} - (B(\frac{1}{2}(q^*_1+q^*_2))_{ij} 
\geq c_{\min} \frac{1}{2} (q^*_1-q^*_2)_k$,
for some indices $i,j,k$.
From  inequality (\ref{eq:ratio-low-q1-q2}) we know that 
$\frac{(q^*_1-q^*_2)_k}{(q^*_1-q^*_2)_{\max}} \geq \alpha^{n}$, 
for all indices $k$.  Thus,
since $(q^*_1-q^*_2)_{\max} = \| q^*_1 - q^*_2\|_\infty$, we have
\begin{eqnarray}
B_1(q^*_1)_{ij} - B_1(\frac{1}{2}(q^*_1+q^*_2))_{ij}  & \geq & c_{\min} \frac{1}{2} (q^*_1-q^*_2)_k \nonumber \\
& \geq & c_{\min} \frac{1}{2} \alpha^n \|q^*_1 - q^*_2\|_\infty \nonumber \\
& \geq & \alpha^{n+1} \frac{1}{2}\|q^*_1 - q^*_2\|_\infty
\label{ineq:low-bound-for-diff-of-jacobs}
\end{eqnarray}

Since $q^*_2 < q^*_1$, $\frac{1}{2}(q^*_1+q^*_2) < q^*_1$.
This combined with Lemma \ref{lem:spectr-at-lfp-is-at-most-1}
together imply  that $\rho(B_1(\frac{1}{2}(q^*_1+q^*_2))) < 1$, and thus
that $(I - B_1(\frac{1}{2}(q^*_1+q^*_2)))^{-1}$ exists and that
$(I - B_1(\frac{1}{2}(q^*_1+q^*_2)))^{-1} = \sum_{i=0}^\infty B_1(\frac{1}{2}(q^*_1+q^*_2)))^i \geq 0$. 
Now we need the following result from \cite{EWY-qbd-2010}:
\begin{lemma}{(Lemma 18 from \cite{EWY-qbd-2010})} \label{lemma19} Let $A\in\mathbb{R}^{n \times n}_{ \geq 0}$ and $b \in \mathbb{R}^n_{\geq 0}$ such that:
$(I-A)^{-1} = \sum_{k=0} ^\infty A^k$ , $(I - A)^{-1} b \leq \mathbf{1}$, and $A$ is an
irreducible nonnegative matrix whose smallest nonzero entry is $c > 0$, and $b \not= 0$ and $p > 0$ is the largest entry of b. Then $\|(I-A)^{-1}\|_\infty \leq \frac{n}{pc^n}$.
\end{lemma}
We will take $A = B_1(\frac{1}{2}(q^*_1+q^*_2))$ and $b= (I - B_1(\frac{1}{2}(q^*_1+q^*_2)))v$ in this Lemma (recall that 
$v$ is the normalized Perron vector of $B_1(q^*_1)$, such that $v_{\max} = 1$). 
We know that $(I - B_1(q^*_1))v \geq 0$. So 
\begin{equation}
b \geq (B_1(q^*_1) - B_1(\frac{1}{2}(q^*_1+q^*_2)))v \geq {\mathbf 0} 
\label{ineq:low-bound-for-b}
\end{equation}
Inequality
(\ref{ineq:low-bound-for-diff-of-jacobs})
 gives us a lower bound for a single entry of the non-negative matrix $(B_1(q^*_1) - B_1(\frac{1}{2}(q^*_1+q^*_2)))$,
namely the $(i,j)$'th entry.  In $(B_1(q^*_1) - B_1(\frac{1}{2}(q^*_1+q^*_2)))v$
this $(i,j)$'th entry is multiplied by a coordinate of $v$, which is at least $v_{\min}$. 
Thus, combining
inequalities (\ref{ineq:low-bound-for-diff-of-jacobs})
and (\ref{ineq:low-bound-for-b}), we 
have $\|b\|_\infty \geq \alpha^{n+1} \frac{1}{2}\|q^*_1 - q^*_2\|_\infty v_{\min}$.
From Corollary \ref{cor:perron-low-bound-ratio}
we have that 
$v_{\min} \geq \alpha^{n}$. 
So $b \geq 0$ and 
$\|b\|_\infty \geq \alpha^{2n+1} \frac{1}{2}\|q^*_1 - q^*_2\|_\infty$. 
Now, by definition, $(I-B_1(\frac{1}{2}(q^*_1+q^*_2)))^{-1}b = v \leq {\mathbf 1}$. 
Since the smallest non-zero entry of 
 $A = B_1(\frac{1}{2}(q^*_1+q^*_2))$ is at least $\alpha$,
and since $\| b \|_\infty \geq \alpha^{2n+1} \frac{1}{2}\|q^*_1 - q^*_2\|_\infty$, 
Lemma \ref{lemma19} now gives that
\begin{equation}
\|(I - B_1(\frac{1}{2}(q^*_1+q^*_2)))^{-1}\|_\infty \leq \frac{2n} {\alpha^{3n+1}\|q^*_1 - q^*_2\|_\infty}
\label{eq:first-bound-on-norm-of-inverse}
\end{equation}

Next suppose that $B_1(x)$ is independent of $x$ (i.e., $P_1(x)$ consists of linear or constant polynomials in $x$). We can thus
write it as $B_1$, a constant, irreducible Jacobian matrix
of $P_1(x)$,  where the MPS $x=P_1(x)$
has an LFP $q^*_1 > 0$.
It must therefore
be the case that $\rho(B_1) < 1$, because we already know from 
Lemma \ref{lem:spectr-at-lfp-is-at-most-1} that for all $z$ such that $0 \leq z < q^*_1$,  we have $\rho(B_1(z)) < 1$,
but $B_1(z)$ is independent of $z$, because $B_1$ is a constant matrix.

Let us apply Lemma 3.3 of \cite{ESY12}, i.e., Lemma \ref{lem:restate-lem-3-3} above,
with $a= q^*_1$, $b=0$, and $P_1(x)$ in place of $P(x)$.
We get  $(B_1) \cdot (q^*_1 - 0) = P_1(q^*_1) - P(0)$.
Multiplying both sides of this equation by $-1$ and then adding $q^*_1$ to both sides, 
we get  $(I - B_1) q^*_1 = P_1(0)$, and thus
$q^*_1 = (I-B_1)^{-1} P_1(0)$. Since $q^*_1 > 0$, we must have that 
$P_1(0) \not= 0$. 
But $P_1(0) \geq 0$. Indeed, $\|P_1(0)\|_\infty \geq 
c_{\min} \min \{1, y_{\min}^2\} \geq \alpha^2$. 
The smallest non-zero entry of $B_1$ 
is at least $c_{\min} \cdot \min \{1, y_{\min}\} \geq \alpha$. 
We now apply Lemma \ref{lemma19} 
to $A := B_1$ and $b := P_1(0)$, where we note 
that $(I-B_1)^{-1}P_1(0) = q^*_1 \leq {\mathbf 1}$.  Lemma \ref{lemma19} thus
gives:

\begin{equation} 
\|(I - B_1(\frac{1}{2}(q^*_1+q^*_2)))^{-1}\|_\infty \leq n \alpha^{-(n+2)}
\label{eq:lin-bound-on-norm-of-inverse}
\end{equation}

Since $\|q^*_1 - q^*_2\|_\infty \leq 1$ ($q^*_1 \leq \mathbf{1}$ and $q^*_2 \geq 0$), and $0 < \alpha \leq 1$, and since $n \geq 1$, the 
upper bound (\ref{eq:first-bound-on-norm-of-inverse}) for
the non-linear case is worse
than the upper bound 
(\ref{eq:lin-bound-on-norm-of-inverse})
for the linear case, so the upper bound
(\ref{eq:first-bound-on-norm-of-inverse}) holds in all cases.

We have shown that $(I - B_1(\frac{1}{2}(q^*_1+q^*_2)))$ is non-singular,
since $\rho( B_1(\frac{1}{2}(q^*_1+q^*_2))) < 1$.
Equation (\ref{eq:iden-deriv-diff-bound}) is thus valid, and taking norms of (\ref{eq:iden-deriv-diff-bound}) yields:
\begin{equation}
\|q^*_1 - q^*_2\|_\infty \leq \|(I - B_1(\frac{1}{2}(q^*_1+q^*_2)))^{-1}\|_\infty \|P_1(q^*_2) - P_2(q^*_2)\|_\infty 
\label{eqn:bound-from-norms} 
\end{equation}
Inserting our upper bound 
(\ref{eq:first-bound-on-norm-of-inverse}) for $\|(I - B_1(\frac{1}{2}(q^*_1+q^*_2)))^{-1}\|_\infty$ gives:
$$\|q^*_1 - q^*_2\|_\infty \leq \frac{2n} {\alpha^{3n+1}\|q^*_1 - q^*_2\|_\infty}
\|P_1(q^*_2) - P_2(q^*_2)\|_\infty$$
We now move the $\|q^*_1 - q^*_2\|_\infty$ terms to the left and take square roots to obtain:

\begin{equation}
\|q^*_1 - q^*_2\|_\infty \leq \sqrt{2n \alpha^{-(3n+1)} \|P_1(q^*_2) - P_2(q^*_2)\|_\infty}
\label{eq:almost-up-bound-on-q1-q2-sc}
\end{equation}
\begin{lemma} If $0 \leq x \leq 1$, then 
$\|P_1(x) - P_2(x)\|_\infty \leq 2 \|P(\textbf{1}, \textbf{1})\|_\infty \|y_1 - y_2\|_\infty$. 
\label{lem:bound-on-diff-P1-P2-in-0-1}
\end{lemma}
\begin{proof}
Since each entry of $P(x,y)$ is a quadratic 
polynomial, for each $b \in \{1,2\}$ 
and each $d \in \{1,\ldots,n\}$,
the $d$'th coordinate, $(P_b(x))_d$, of $P_b(x) = P(x,y_b)$  
has the form $$\sum_{i,j} a_{d,i,j} x_i x_j + \sum_{i,j} c_{d,i,j} y_{b,i} y_{b,j}  + 
\sum_{i,j} c'_{d,i,j} x_i y_{b,j} + \sum_k a'_{d,k} x_k + \sum_{k} c''_{d,k} y_{b,k} +   c'''_d$$
where $y_{b,j}$ refers to the $j$'th coordinate of the $m$-vector $y_b$,
and where all the coefficients $a_{d,i,j}$, $c_{d,i,j}$, $c'_{d,i,j}$, $c''_{d,k}$ and $c'''_d$, 
are non-negative.
Also, recall $0 < y_1 \leq {\mathbf 1}$ and $0 \leq y_2 \leq y_1$. 
Thus, 
\vspace*{0.1in}

\noindent $\| P_1(x) - P_2(x) \|_\infty  $
\begin{eqnarray*}
& = &   \max_d
 \sum_{i,j} c_{d,i,j} (y_{1,i}y_{1,j} - y_{2,i} y_{2,j}) +  \sum_{i,j} c'_{d,i,j} x_i (y_{1,j} - y_{2,j}) + \sum_k  c''_{d,k} (y_{1,k} - y_{2,k})\\
& \leq & 
\max_d \sum_{i,j}  c_{d,i,j}  ( (y_{1,i} - y_{2,i})  +  (y_{1,j}  - y_{2,j})) +  \sum_{i,j} c'_{d,i,j} (y_{1,j}- y_{2,j}) + \sum_k c''_{d,k} (y_{1,k} - y_{2,k})\\ 
& \leq & \max_d \sum_{i,j} 2 \cdot c_{d,i,j} \cdot \| y_1 - y_2 \|_\infty +  \sum_{i,j} 
c'_{d,i,j} \cdot  \| y_1 - y_2 \|_\infty  + \sum_k c''_{d,k} \| y_1 - y_2 \|_\infty \\
& = &  (\max_d \sum_{i,j} 
2 c_{d,i,j}  + c'_{d,i,j} + c''_{d,k}) \cdot \| y_1 - y_2 \|_\infty\\  
& \leq & 
2 \| P(1,1) \|_{\infty} \| y_1 - y_2 \|_\infty 
\end{eqnarray*} 
\qed   
\end{proof}

\noindent Combining (\ref{eq:almost-up-bound-on-q1-q2-sc}) and
Lemma
\ref{lem:bound-on-diff-P1-P2-in-0-1}, we have, 
$$\|q^*_1 - q^*_2\|_\infty \leq \sqrt{4n \alpha^{-(3n+1)} \|P(\textbf{1}, \textbf{1})\|_\infty \|y_2 - y_1\|_\infty}$$
which completes the proof of the
first inequality of part ({\bf 1.}) of Theorem 
\ref{scmps}. 

We show the second inequality (\ref{eq:bound-on-lin-diff-of-lfps}) of part ({\bf 1.}) in the next lemma.

\begin{lemma} If $B_1(x)$ is a constant matrix, i.e. $x=P_1(x)$ is linear,
$$ \|q^*_1 - q^*_2\|_\infty \leq 2n \alpha^{-(n+2)} \|P(\mathbf{1}, \mathbf{1})\|_\infty \|y_1 - y_2\|_\infty $$
\label{lem:linear-scc} \end{lemma}
\begin{proof}
If $q^*_1=q^*_2$, then the result is trivial. 
So we assume that $q^*_1 \not= q^*_2$.
In the proof of the first inequality, under the assumption that $B_1(x)$ is a constant, 
we obtained equation (\ref{eq:lin-bound-on-norm-of-inverse}). 
So, under the assumptions of this Lemma, equation (\ref{eq:lin-bound-on-norm-of-inverse}) is valid,
and we can substitute the bound (\ref{eq:lin-bound-on-norm-of-inverse}) into the equation (\ref{eqn:bound-from-norms}) instead. This gives

$$ \|q^*_1 - q^*_2\|_\infty \leq n \alpha^{-(n+2)} \|P_1(q^*_2) - P_2(q^*_2)\|_\infty $$
Again, Lemma \ref{lem:bound-on-diff-P1-P2-in-0-1} gives a bound on $\|P_1(q^*_2) - P_2(q^*_2)\|_\infty $. Substituting this gives:

$$ \|q^*_1 - q^*_2\|_\infty \leq n \alpha^{-(n+2)} 2 \|P(\mathbf{1}, \mathbf{1})\|_\infty \|y_1 - y_2\|_\infty $$
\qed
\end{proof}

We will next establish part ({\bf 2.}) of Theorem \ref{scmps}.
Let us first prove that, starting from $x^{[0]} := 0$, 
all the iterations of R-NM, applied
to $x=P_2(x)$ are defined.

We firstly note that if $0 \leq x^{[k]} \leq q^*_2$ and
$\rho(B_2(x^{[k]})) < 1$, then $\mathcal{N}_{P_2}(x^{[k]})$ is
well-defined and $0 \leq x^{[k+1]} \leq q^*_2$.  If
$\rho(B_2(x^{[k]})) < 1$, then by Lemma
\ref{lem:inverse-of-sub-1-spec-exists}, $(I- B_2(x^{[k]}))$ is
non-singular and so $\mathcal{N}_{P_2}(x^{[k]})$ is
well-defined. Lemma \ref{lem:inverse-of-sub-1-spec-exists} also gives
that $(I- B_2(x^{[k]}))^{-1} = \sum_{i=0}^\infty B_2(x^{[k]})^i \geq
0$. Lemma \ref{newton-from-esy12} yields that:
$$q^*_2 -  \mathcal{N}_{P_2}(x^{[k]}) = (I-B_2(x^{[k]}))^{-1}\frac{B_2(q^*_2) - B_2(x^{[k]})}{2}(q^*_2 - x^{[k]})$$
Note that $(q^*_2 - x^{[k]}) \geq 0$, thus that
$B_2(q^*_2) - B_2(x^{[k]}) \geq 0$, and we have just shown that $(I-B_2(x^{[k]}))^{-1} \geq 0$. So 
all the terms on the right of the above equation are non-negative, and thus $q^*_2 -  \mathcal{N}_{P_2}(x^{[k]}) \geq 0$. That is $q^*_2 \geq  \mathcal{N}_{P_2}(x^{[k]})$. $x^{[k+1]}$ is defined by rounding down  $\mathcal{N}_{P_2}(x^{[k]})$
and maintaining non-negativity, 
thus for all coordinates $i$, either $x^{[k+1]}_i=0$,
in which case trivially we have $x^{[k+1]}_i = 0 \leq (q^*_2)_i$,
or else
 $0 \leq  x^{[k+1]}_i \leq \mathcal{N}_{P_2}(x^{[k]})_i \leq (q^*_2)_i$.
Thus $x^{[k+1]} \leq q^*_2$.

What is still missing is to show that $\rho(B_2(x^{[k+1]})) < 1$. If we can show this
then by an easy induction, for all $k$, $\mathcal{N}_{P_2}(x^{[k]})$ is well-defined and $0 \leq x^{[k]} \leq q^*_2$. 
We will prove $\rho(B_2(x^{[k+1]})) < 1$ by considering separately the cases where $P_1(x)$ contains non-linear
or only linear
polynomials.

\begin{lemma} \label{non-linear-sc-strict} If $x=P(x)$ is a strongly-connected quadratic MPS 
with $n$ variables, with LFP $q^* > 0$,
and there is some non-linear quadratic term in some polynomial $P_i(x)$, 
then if $0 \leq z < q^*$, then $\mathcal{N}_{P}(z)$ is defined and $\mathcal{N}_{P}(z) < q^*$. \end{lemma}
\begin{proof} Lemma \ref{lem:spectr-at-lfp-is-at-most-1} tells us
that $\rho(B(q^*)) \leq 1$. Non-linearity of $P(x)$ means that $B(x)$ does depend on $x$. 
That is, some entry of $B(x)$ contains  a term of the form $c x_i$ for some $x_i$ with $c > 0$. 
So $B(z) \not= B(q^*)$, and $B(z) \leq B(q^*)$ since $B$ is monotone. 
Since $x=P(x)$ is strongly-connected and $q^* > 0$, 
Lemma \ref{lem:spectr-at-lfp-is-at-most-1}
yields that $\rho(B(z)) < 1$. By Lemma \ref{lem:inverse-of-sub-1-spec-exists}, $(I-B(z))$ is non-singular and so the Newton iterate ${\mathcal N}_{P}(z)$ is well-defined. 
Consider the equation given by Lemma \ref{newton-from-esy12}:
$$q^* -  \mathcal{N}_{P}(z) = (I-B(z))^{-1}\frac{B(q^*) - B(z)}{2}(q^* - z)$$
We know that $q^* - z > 0$, and thus $B(q^*) - B(z) \geq 0$. 
Since $\rho(B(z)) < 1$, by Lemma \ref{lem:inverse-of-sub-1-spec-exists}, $(I-B(z))^{-1} = \sum_{k=0}^\infty B(z)^k \geq 0$. This and 
Lemma \ref{newton-from-esy12} is already enough 
to yield that $q^* - \mathcal{N}_{P}(z) \geq 0$, and we just need to show that this is a strict inequality. 

We first show that if $P_i(x)$ contains a term of degree 2, then $(\frac{B(q^*) - B(z)}{2}(q^* - z))_i > 0$. This term of degree 2 must be of the form $cx_jx_k$ for some $j,k$. Then $B(x)_{i,j}$ has a term $cx_k$ with $c > 0$ and so $(B(q^*) - B(z))_{i,j} \geq c(q^* - z)_k$. But then $(\frac{B(q^*) - B(z)}{2}(q^* - z))_i \geq c(q^* - z)_k(q^* - z)_j > 0$.

Now we will show that for all $i \in \{1,\ldots,n\}$, $(q-\mathcal{N}_{P}(z))_i > 0$. If $P_i(x)$ 
contains a term of degree 2, then we have just shown that $(\frac{B(q^*) - B(z)}{2}(q^* - z))_i > 0$.
 But $(I-B(z))^{-1} = \sum_{k=0}^\infty B(z)^k \geq I$. 
So $(q-\mathcal{N}_{P}(z))_i \geq (\frac{B(q^*) - B(z)}{2}(q^* - z))_i > 0$. 
 If $P_i(x)$ does not contain a term of degree 2, there must be some other $x_j$ with $P_j(x)$ containing a term of degree 2 and, since $x=P(x)$ is strongly-connected, $x_i$ depends on $x_j$, possibly indirectly.
 That is, there is a sequence of variables $i_0$,$i_1,...,i_l$ with $l < n$, $i_0=i$,$i_l=j$, and for each $0 < m \leq l$,$x_{i_m}$ appears in a term of $P(x)_{i_{m-1}}$.
 Let $k$ be the the least integer such that $P(x)_{i_k}$ contains a term of degree 2.
 Then if $0 < m \leq k$, $x_{i_m}$ appears in a degree 1 term in $P(x)_{i_{m-1}}$, that is one of the form $c_mx_m$ with $c_m > 0$. So $B(x)_{i_{m-1},i_m}$ contains the constant term $c_m > 0$. So $B(z)_{i_{m-1},i_m} \geq c_m > 0$.
 So $B^k(z)_{i,i_k} \geq \prod_{m=0}^{k-1} B(z)_{i_m,i_{m+1}} \geq  \prod_{m=0}^{k-1} c_m > 0$.
 Since $P(x)_{i_k}$ contains a term of degree 2, from above $(\frac{B(q^*) - B(z)}{2}(q^* - z))_{i_k} > 0$.
 So $(B^k(z)\frac{B(q^*) - B(z)}{2}(q^* - z))_i > 0$.
 But $q^* -  \mathcal{N}_{P}(z) = (I-B(z))^{-1}\frac{B(q^*) - B(z)}{2}(q^* - z) = (\sum_{m=0}^\infty B^m(z)) \frac{B(q^*) - B(z)}{2}(q^* - z) \geq B^k(z) \frac{B(q^*) - B(z)}{2}(q^* - z)$. So $(q^*-\mathcal{N}_{P}(z))_i > 0$ for all $i$, as required.
\qed   \end{proof}
We will only actually need to apply
Lemma \ref{non-linear-sc-strict} in the case when  $q^*_2=q^*_1$ and $x=P_1(x)$ is non-linear. 

Suppose that $q^*_1 = q^*_2$ and some polynomial in $P_1(x)$ is non-linear in $x$. 
We claim that then $P_1(x) \equiv
P_2(x)$.  That is, for all those variables in $y$,  say $(y)_j$, that actually appear in
some polynomials in $P(x,y)$, it must be the case that $(y_1)_j= (y_2)_j$.  
Otherwise, if there is some variable 
$(y)_j$ with $(y_2)_j < (y_1)_j$ such that $(y)_j$ appears in $P_i(x,y)$, then 
$(P_2(q^*_1))_i
= (P(q^*_1,y_2))_i < P(q^*_1,y_1))_i = (q^*_1)_i$, so $q^*_1$ is not
a fixed point of $P_2(x)$, contradicting that $q^*_1=q^*_2$.  Thus if $x=P_1(x)$ is non-linear 
and $q^*_1 = q^*_2$ then
$x=P_2(x)$ is also non-linear and $q^*_2 = q^*_1 > 0$, so we can use
Lemma \ref{non-linear-sc-strict},  which shows that if $0 \leq x^{[k]} < q^*_2$,
then $\mathcal{N}_{P_2}(x^{[k]}) < q^*_2$ and so $0 \leq x^{[k+1]} < q^*_2 \leq q^*_1$. 
Since $x^{[k+1]} < q^*_1$,  we have $\rho(B_1(x^{[k+1]})) < 1$.
Since $B_2(x^{[k+1]}) \leq B_1(x^{[k+1]})$, 
we also have $\rho(B_2(x^{[k+1]})) < 1$.

This leaves us with two cases remaining to show that all Newton iterates exist:
first, 
the case where $x=P_1(x)$ is linear or constant, and second, the case where $x=P_1(x)$ is non-linear and $q^*_2 \not= q^*_1$. 
Recall that it is sufficient to show that $\rho(B_2(x^{[k]})) < 1$ for all iterates
in order to show that all R-NM iterates exist. 
It thus suffices to show that in these cases for any $0 \leq z \leq q^*_2$, $\rho(B_2(z)) < 1$. 

For the first case, suppose that $x=P_1(x)$ is linear. Then $B_1(x)$ is a constant
matrix.  Thus $B_1(z) = B_1(0)$ for all $0 \leq z$. But
Lemma \ref{lem:spectr-at-lfp-is-at-most-1} tells us that, since $0 < q^*_1$, $\rho(B_1(0)) < 1$. 
Thus $\rho(B_1(z)) < 1$ for all $0 \leq z \leq q^*_2$. Since $0 \leq B_2(z) \leq B_1(z)$,
we have $\rho(B_2(z)) < 1$ for all $0 \leq z \leq q^*_2$.

For the second case, suppose that $q^*_2 \not= q^*_1$ and that $x=P_1(x)$ is non-linear,
and thus $B_1(x)$ depends on $x$. 
Then we have previously argued that $q^*_2 < q^*_1$  (see inequality (\ref{eq:q2-less-q1})). 
But then $B_1(q^*_2) \not= B_1(q^*_1)$. For any $0 \leq z \leq q^*_2$, 
$B_2(z) \leq B_2(q^*_2) \leq B_1(q^*_2) \leq B_1(q^*_1)$ but because $B_1(q^*_2) \not= B_1(q^*_1)$, 
we have $B_2(z) \not= B_1(q^*_1)$. 
But $B_1(q^*)$ is irreducible, and 
Lemma \ref{lem:irr-non-neg-spec-lower} then tells us that $\rho(B_2(z)) < \rho(B_1(q^*_1))$. 
But we know, by Corollary \ref{cor:perron-cone}, that $\rho(B_1(q^*_1) \leq 1$. So  $\rho(B_2(z)) < 1$.

Thus the R-NM iterations applied to $x=P_2(x)$ are defined in all cases, and yield iterates
$0 \leq x^{[k]} \leq q^*$, for all $k \geq 0$.
We can now prove the upper bound on the rate of
convergence for R-NM
applied to $x = P_2(x)$.

\begin{lemma} \label{roundedNewton} Suppose an MPS, $x = P(x)$, with $n$ variables
has LFP $0 \leq q^* \leq 1$, and for some $n$-vector $v > 0$ we have
$B(q^*)v \leq v$.   
Suppose we perform $g \geq h-1$ iterations of R-NM with 
parameter $h \geq 2 + \lceil \log \frac{v_{\max}}{v_{\min} \cdot \epsilon} \rceil$ 
on the MPS $x=P(x$),  
and suppose that for all $k \geq 0$,
every iteration $x^{[k]}$ is defined and $0 \leq x^{[k]} \leq q^*$.  Then
$\|q^* - x^{[g]}\|_\infty 
\leq \epsilon$. \end{lemma}
\begin{proof}
By induction on $k$,
we claim that $\forall k \geq 0$, 
$q^*-x^{[k]} \leq (2^{-k} + 2^{-h+1}) \frac{1}{v_{\min}}v$. 
Note that this would indeed yield the Lemma:
for all $k$  $0 \leq x^{[k]} \leq q^*$.
and the claim would yield
$q^*-x^{[g]} \leq (2^{-h+1} + 2^{-h+1}) \frac{1}{v_{\min}}v \leq 2^{-\log \frac{v_{\max}}{v_{\min} \cdot \epsilon}}  \frac{1}{v_{\min}}v
= \epsilon \frac{1}{v_{\max}} v \leq \epsilon {\mathbf 1}$.

It remains to prove by induction on $k \geq 0$ that
$q^*-x^{[k]} \leq (2^{-k} + 2^{-h+1}) \frac{1}{v_{\min}}v$. 
This is true for 
$k = 0$, because $q^* \geq 0 = x^{[0]}$, and $q^* - x^{[0]} = q^* \leq 1 \leq \frac{1}{v_{\min}}v$.

Lemma \ref{cone} then gives that $q^* - {\mathcal N}_P(x^{[k]}) \leq (2^{-(k+1)} + 2^{-h}) \frac{1}{v_{\min}}v$. 
 Now, by definition of $x^{[k+1]}$,  
${\mathcal N}_P(x^{[k]}) - x^{[k+1]} \leq 2^{-h} \mathbf{1} \leq 2^{-h} \frac{1}{v_{\min}}v$. 
 So $q^* - x^{[k+1]} \leq  (2^{-(k+1)} + 2^{-h+1})\frac{1}{v_{\min}}v$ as required. 
Thus $q^*- x^{[h - 1]} \leq 2^{-h+2}\frac{1}{v_{\min}}v \leq \frac{\epsilon}{v_{\max}} v \leq 
\epsilon \mathbf{1}$.
\qed   \end{proof}
To use Lemma \ref{roundedNewton} to get a bound on
using R-NM on
$x=P_2(x)$ to compute $q^*_2$, note that 
because  $0 \leq B_2(q^*_2) \leq B_1(q^*_1)$,
the Perron vector $v > 0$ of $B_1(q^*_1)$, which satisfies $B_1(q^*_1) v \leq v$, must
also satisfy $B_2(q^*_2)v \leq v$.

Thus, we just need  to perform $g \geq h-1$ iterations of R-NM on $x=P_2(x)$, with parameter
$h \geq 2 + \mathrm{log} \frac{v_{\max}}{v_{\min}\epsilon} \geq 2 + \mathrm{log} \alpha^{-n}\epsilon^{-1}$
in order to obtain that $\| q^*_2 - x^{[h-1]}\|_\infty \leq \epsilon$.
This completes the proof of Theorem \ref{scmps}.
\qed   \end{proof}

\noindent {\bf Corollary \ref{cor:single-scmps}.} 
{\em Let  $x=P(x)$ be a strongly connected MPS with $n$ 
variables, and with LFP $q^*$ where $0 < q^* \leq 1$.
Let
 $\alpha = {\min} \{1,c_{\min}\} \frac{1}{2}q^*_{\min}$, where 
$c_{\min}$ is the smallest non-zero constant or coefficient 
of any monomial in $P(x)$. 

Then for all $0 < \epsilon < 1$,  
if we use $g \geq h-1$ iterations of  R-NM 
with parameter 
$h \geq \lceil 2 +  n \log \frac{1}{\alpha} + \log \frac{1}{\epsilon} \rceil$
applied to the MPS, $x=P(x)$,
starting at $x^{[0]} := \mathbf{0}$, 
then the iterations are all defined, and
$\|q^* - x^{[g]} \|_\infty \leq \epsilon$.
}

\begin{proof} This is just a trivial application of Theorem \ref{scmps}, part {\bf 2.},  where
we define $y$ to be a dummy variable of dimension $m=1$, and
we define $y_1 = y_2 = y_{\min} = 1$, and where we define the $n$-vector
of monotone polynomials
$P(x,y)$, by replacing all constant terms $c > 0$ in every polynomial in $P(x)$ 
by $c y$.   In this case, note that $P_1(x) = P_2(x) = P(x)$, and that
since $y_{\min} = 1$, the $\alpha$ defined in the statement of this corollary is the same $\alpha$ as in Theorem
\ref{scmps}.\qed
\end{proof}

\subsection{Proof of Theorem \ref{gen}}

\noindent {\bf Theorem \ref{gen}.}
{\em If $x=Q(x)$ is an MPS with $n$ variables, with LFP solution $q^* > 0$,  
if $c'_{\min}$ is the least positive coefficient of any monomial in $Q(x)$, 
then R-DNM 
with rounding parameter $h'$, and using $g'$ iterations per nonlinear SCC (and one for linear), gives an approximation 
$\tilde{q}$ such that $\|q^* - \tilde{q}\|_\infty \leq \epsilon'$, where
$$g'= 2 + \lceil \; 2^f \cdot (\log (\frac{1}{\epsilon'}) + d \cdot (2u   +  \log (\alpha'^{-(4n+1)}) + \log(16n) + \log (\| Q({\mathbf 1}) \|_\infty)) \; ) \ \rceil $$ 
and 
$h' = g'+1 - u$, 
where $u = \max \{ 0, \lceil \log q^*_{\max} \rceil \}$,  $d$ is the maximum depth of SCCs in the DAG $H_Q$ of SCCs of the dependency graph of 
$x=Q(x)$, $f$ is the nonlinear depth of $H_Q$, and $\alpha' = 2^{-2u} \min \{1,c'_{\min}\} \min \{1, \frac{1}{2}q^*_{\min}\}$.}

\begin{proof} If $q^*_{\max} \leq 1$, then Theorem \ref{gen<=1} gives
  this immediately. So we assume that $q^*_{\max} > 1$. $u$ is chosen so that $2^u \geq q^*_{\max}$. 
We rescale and use
  Lemma \ref{scale} with scaling parameter $c = 2^u$.  This yields the
``rescaled''
  MPS  $x = 2^{-u}Q(2^ux)$, which has LFP
  $p^* = 2^{-u}q^* \leq \mathbf{1}$.

So we can apply
  Theorem \ref{gen<=1} to this rescaled MPS $x=P(x)$,
where $P(x) \equiv 2^{-u}Q(2^u x)$, and 
letting $\epsilon := 2^{-u}\epsilon'$.
Then Theorem \ref{gen<=1} gives us the needed number of iterations 
$g$ and
the rounding parameter $h = g+1$, needed
to obtain an approximation $\tilde{p}$ of the LFP $p^* = 2^{-u}q^*$, such that 
$\| \tilde{p} - p^* \|_\infty \leq \epsilon$.

In the bounds specified for Theorem  \ref{gen<=1}
for $g$ and $h$,  
in place of 
$q^*_{\min}$ we get $p^*_{\min}=2^{-u} q^*_{\min}$,
and in place of $c_{\min}$ we get $2^{-u} c'_{\min}$.
Thus $\alpha$ becomes  the $\alpha'$ we have specified in the statement
of this theorem.
Furthermore, the $\| P({\mathbf 1}) \|_\infty$ appearing in Theorem \ref{gen<=1}
is now $\| 2^{-u}Q(2^u {\mathbf 1}) \|_\infty$, 
but it is easy to verify that for a quadratic MPS, 
$\| 2^{-u}Q(2^u {\mathbf 1}) \|_\infty \leq 2^u  \| Q({\mathbf 1}) \|_\infty$.

Theorem \ref{gen<=1} tells us that if we use R-DNM on $x=P(x)$ for $g$
iterations per nonlinear SCC and a precision of $h=g+1$ bits, we will obtain an
approximation $\tilde{p}$ to the LFP $p^*$ of $x=P(x)$
with $\|\tilde{p} - p^*\|_\infty \leq \epsilon$ 
provided that $h \geq \lceil 3 + 2^f \cdot ( \log (\frac{1}{\epsilon}) + d \cdot (\log (\alpha^{-(4n+1)}) + \log(16n) + \log (\| P({\mathbf 1}) \|_\infty)) ) \rceil$. 
This condition is satisfied if we take $g=g'$ and $h=g'+1$ because:

\begin{eqnarray*}  &  & \lceil 3 + 2^f \cdot ( \log (\frac{1}{\epsilon}) + d \cdot (\log (\alpha^{-(4n+1)}) + \log(16n) + \log (\| P({\mathbf 1}) \|_\infty)) ) \rceil \\
& \leq & 3 + 2^f (\log (\frac{1}{2^{-u}\epsilon'}) + d( \log (\alpha'^{-(4n+1)}) + \log(16n) + \log (2^u\| Q({\mathbf 1}) \|_\infty) )) \rceil \\
& = & 3 + 2^f(u + \log (\frac{1}{\epsilon'}) + d( \log (\alpha'^{-(4n+1)}) + \log(16n) + u+\log (\| Q({\mathbf 1}) \|_\infty) )) \rceil \\
& \leq & g' + 1 =h\end{eqnarray*}
Thus, applying R-DNM on $x=P(x)$ with parameters $g=g'$ and
$h=g'+1$ yields an approximation
 $\tilde{p}$ to the LFP $p^*$ of $x=P(x)$
with $\|\tilde{p} - p^*\|_\infty \leq \epsilon$ or, in terms of the original MPS,  $\|\tilde{p} - 2^{-u}q^*\|_\infty \leq 2^{-u}\epsilon'$.

To obtain Theorem \ref{gen}, we now show that if we
apply R-DNM to $x=Q(x)$
with LFP $q^*$, 
using rounding parameter $h'$ and using  
$g'$ iterations per nonlinear SCC (where $h'$ and $g'$ were specified
in the statement of the Theorem), we will obtain an
approximation $\tilde{q}$ to $q^*$ that satisfies
$\tilde{q} = 2^{u}\tilde{p}$.
This would then give us that 
$\|q^* - \tilde{q}\|_\infty =  \| 2^{u} p^* - 2^u \tilde{p} \|_{\infty} = 
2^u \| p^* - \tilde{p} \|_{\infty} \leq 2^u \epsilon = \epsilon'$,
which is what we want to prove.  

Since we are using the decomposed Newton's method, we will show 
that $\tilde{q}_S = 2^{u}\tilde{p}_S$ for every SCC $S$ by induction on 
the depth of the SCC $S$. 
Suppose that for the variables $D(S)$ that $S$ depends on (if any), we have 
that $\tilde{q}_{D(S)} = 2^{u}\tilde{p}_{D(S)}$.
If we call the $k$th iterate of R-NM applied to $x_S=P_S(x_S, \tilde{p}_{D(S)})$ with parameter $h$, $x^{[k]}$ and the $k$th iterate of R-NM applied to $x_S=Q_S(x_S, \tilde{q}_{D(S)} )$ with parameter $h'$, $x'^{[k]}$, then we aim to show by induction on $k$ that $x'^{[k]} = 2^u x^{[k]}$.

The base case is $x'^{[0]} = 0 = 2^u x^{[0]}$.
By abuse of notation, we will call the Newton iterate of 
$x_S=P_S(x_S, \tilde{p}_{D(S)})$, $\mathcal{N}_P(x_S)$ and that of $x_S=Q_S(x_S, \tilde{q}_{D(S)} )$, 
$\mathcal{N}_Q(x_S)$. Note that because we assume that $\tilde{q}_{D(S)} = 2^u \tilde{p}_{D(S)}$, $x_S=P_S(x_S, \tilde{p}_{D(S)})$ is the result of scaling $x_S=Q_S(x_S, \tilde{q}_{D(S)} )$ using $c=2^u$. 
So Lemma \ref{scale} (iii) yields that $\mathcal{N}_P(x_S) = 2^{-u}{\mathcal N}_Q(2^ux_S)$. If $x'^{[k]} = 2^u x^{[k]}$, then $\mathcal{N}_Q(x'^{[k]}) = 2^u \mathcal{N}_P(x^{[k]})$.

If $(\mathcal{N}_P(x^{[k]}))_i \leq 0$, we would set $x^{[k+1]}_i:=0$. If so, 
$\mathcal{N}_Q(x'^{[k]})_i = 2^u \mathcal{N}_P(x^{[k]})_i \leq 0$, so we would set $x'^{[k+1]}_i:=0$.

If $(\mathcal{N}_P(x^{[k]}))_i > 0$, we set $x^{[k+1]}_i$ to be the result of rounding $(\mathcal{N}_P(x^{[k]}))_i$ down to a multiple of $2^{h}$. 
But then $\mathcal{N}_Q(x'^{[k]}) = 2^u \mathcal{N}_P(x^{[k]}) > 0$ and 
we would set $x'^{[k+1]}_i$ to be the result of rounding $(\mathcal{N}_Q(x'^{[k]}))_i$ down to a multiple of $2^{-h'}$. Note that $h' = h - u$. So the result of rounding $2^u (\mathcal{N}_P(x^{[k]}))_i$ down to a multiple of $2^{-h'}$ is just $2^u$ times the result of  rounding $(\mathcal{N}_P(x^{[k]}))_i$ down to a multiple of $2^{-h}$. So $x'^{[k+1]} = 2^u x^{[k+1]}$.

This completes the induction showing that $x'^{[k]} = 2^u x^{[k]}$ for all $k \geq 0$. Note that $g=g'$. So $\tilde{q}_{S} = x'^{[g']} = 2^{u} x^{[g]} = 2^{u}  \tilde{p}_{S}$. This in turn completes the induction on the SCCs, showing that $\tilde{q} = 2^{u}\tilde{p}$, which completes the proof.\qed\end{proof}

\subsection{Proof of Theorem \ref{thm:qbd-termin-p-time}}

\noindent {\bf Theorem 
\ref{thm:qbd-termin-p-time}.}
{\em 
Let $x=P(x)$ be the MPS associated with p1CA, $M$, 
let $r$ denote the number of
control states of $M$, and let $m$ denote the maximum number
of bits required to represent the numerator and denominator of
any positive rational transition probability in $M$.

Apply R-DNM, including rounding down linear SCCs, to the MPS $x=P(x)$,
using rounding parameter
$$h := 8mr^7 + 
2mr^5 + 9r^2 + 3 + \lceil 2 \log \frac{1}{\epsilon} \rceil$$
and such that for each non-linear SCC we perform $g = h-1$ iterations,
whereas for each linear SCC we only perform $1$ R-NM iteration.  

This algorithm computes an approximation $\tilde{q}$ to $q^*$, such that $\|q^* - \tilde{q}\|_\infty < 
\epsilon$.
The algorithm runs in time polynomial in $| M|$ and 
$\log \frac{1}{\epsilon}$, in the standard Turing model of computation.
}

\begin{proof}
We apply Theorem \ref{gen<=1}, 
which tells us that R-DNM  with parameter 
\begin{equation}
h \geq \left\lceil 3 + 2^f \cdot (\; \log (\frac{1}{\epsilon}) +  d \cdot (\log (\alpha^{-(4n+1)}) + \log(16n) + \log (\| P({\mathbf 1}) \|_\infty)) \; ) \right\rceil \label{eq:qbd-needsub}
\end{equation}
using $g = h-1$ iterations for every SCC, gives an 
approximation $\tilde{q}$ to $q^*$ with $\tilde{q} \leq q^*$ and such that $\|q^* - \tilde{q}\|_\infty \leq \epsilon$.
Here $f \leq 1$ since there is at most $1$ non-linear SCC in any path through 
the dependancy graph. 
Furthermore, $n=r^2$ since the variables in $x$ are indexed by two states $x_{uv}$. Also, $d \leq n$, and so $d \leq r^2$.  
Also, $c_{\min} \geq 2^{-m}$ and so by Lemma \ref{lem:low-bound-probs-in-qbds}, $q^*_{\min} \geq 2^{-mr^3}$. So 
$\alpha \geq 2^{-(mr^3+1)}$. 
To show that $\|P(\mathbf{1})\|_\infty  \leq r$, by equation (\ref{eq:xuv}), 
$P(\mathbf{1})_{uv} =  p^{(-1)}_{uv} + (\sum_{w \in V} p^{(0)}_{uw}) + \sum_{y \in V} p^{(1)}_{uy} r \leq r$. Plugging all this into equation 
(\ref{eq:qbd-needsub}),  we get:
$h \geq \lceil 3 + 2 \cdot (\log (\frac{1}{\epsilon}) + r^2 \cdot ((4r^2 + 1)(mr^3 + 1) + \log(16r^2) + \log r \rceil$.
Noting that $\log(16r^2) + \log r =  \log(16 r^3)$, and noting
that when $r \geq 1$, $\log (16r^3) \leq 4r$, we have:
$$h \geq 3 + 8mr^7 + 2mr^5 + 9r^4 + \lceil 2 \cdot \log (\frac{1}{\epsilon}) \rceil$$

Note that the rounding parameter $h$ and the number
of iterations $g=h-1$ are both polynomials in the encoding size of the p1CA,
and in $\log \frac{1}{\epsilon}$.
Thus   
each iteration of R-DNM can be computed in polynomial time, 
and we only do polynomially many iterations. 
Thus the entire computation of $\tilde{q}$ can be carried out in P-time in the Turing model of computation.\qed
\end{proof}

\subsection{Proof of Theorem \ref{thm:qbd-mcing}}

{\bf Theorem \ref{thm:qbd-mcing}.}
{\em Given a p1CA, $M$, with states labeled from 
an alphabet $\Sigma$, and
with a specified initial control state $v$, and given an $\omega$-regular
property $L(B) \subseteq \Sigma^\omega$, 
which is specified by a non-deterministic 
B\"{u}chi automaton, $\mathcal{B}$,  let $Pr_{M}(L(B))$ denote the probability that a run of
$M$ starting at configuration $(v,0)$ generates an $\omega$-word
in $L(B)$.  There is an algorithm that, for any $\epsilon > 0$,
computes an additive 
$\epsilon$-approximation, $\tilde{p} \geq 0$, of $Pr_{M}(L(B))$,
i.e., with $| Pr_{M}(L(B)) - \tilde{p} | \leq \epsilon$.
The algorithm runs in time polynomial in $|M|$,
$\log \frac{1}{\epsilon}$, and $2^{|\mathcal{B}|}$,   
in the standard Turing model of computation.}

\begin{proof}[sketch]
By Theorem \ref{thm:qbd-termin-p-time}, we know 
we can compute termination probabilities $q^*$ for a p1CA, $M$, 
with additive error $\epsilon > 0$ in time polynomial
in $|M|$ and $\log \frac{1}{\epsilon}$.

Let us first observe that if we do not insist on
having the $\omega$-regular property specified by a non-deterministic 
B\"{u}chi automaton  $\mathcal{B}$, and instead assume
it is specified by a deterministic Rabin automaton $R$, 
then the analogous theorem follows
immediately as a corollary of 
Theorem \ref{thm:qbd-termin-p-time} and results 
established by
Brazdil, Kiefer, and Kucera in \cite{BKK11}.
Specifically, in \cite{BKK11} it was shown that,
given
a p1CA, $M$, and a 
deterministic Rabin automaton, $R$,
and given $\epsilon > 0$,  
there is an algorithm that, firstly,
decides in P-time whether $P_M(L(R)) > 0$, and
if so computes a value $\tilde{p}$ which 
approximates $P_M(L(R))$ with {\em relative error}  $\epsilon > 0$,
i.e., such that $|P_M(L(R)) - \tilde{p} |/  P_M(L(R))  < \epsilon$,
and the algorithm runs in time polynomial in $|M|$, $|R|$, and
$\log \frac{1}{\epsilon}$,  {\em in the unit-cost RAM model of
computation}.

The first observation we make is that,
the results in \cite{EWY-qbd-2010} and \cite{BKK11}
together imply that
for p1CAs there is no substantial difference
in complexity between relative and absolute approximation,
because the probabilities $P_M(L(R))$ 
can be bounded away from zero by $1/2^{poly(|M|,|R|)}$
if it is not equal to zero (which can be detected in P-time).
Thus, computing $P_M(L(R))$ with given relative error $\epsilon > 0$ 
is P-time
equivalent to computing $P_M(L(R))$ with $\epsilon$ absolute error.

Secondly, a close inspection of \cite{BKK11} shows that
the {\em only} use made in their entire paper
of the unit-cost RAM model of computation
is for the purpose of computing termination
probabilities for p1CAs, and specifically because they 
directly invoke the earlier result from \cite{EWY-qbd-2010} which showed
that termination probabilities $q^*$ for 
a p1CA can be $\epsilon$-approximated 
in polynomial time in the unit-cost RAM model.
Thus, the only thing needed in order to 
obtain an absolute error 
$\epsilon$-approximation of $P_M(L(R))$ in
P-time in the standard Turing model of computation
is to appeal instead to
Theorem \ref{thm:qbd-termin-p-time} of this paper for computation
of termination probabilities in P-time in the standard Turing
model, and apply the rest of the construction in \cite{BKK11}.

Next, let's first note that we can of course use Safra's construction 
to 
convert any non-deterministic
B\"{u}chi automaton $\mathcal{B}$ to a deterministic Rabin
automaton of size $2^{O(|B| \log |B|)}$.
So, obtaining a complexity bound that is polynomial in
$2^{|B| \log |B|}$ is no more difficult.

Let us now very briefly sketch 
why one can in fact obtain the (slightly) better complexity bound,
polynomial in $2^{|B|}$, by combining 
prior results regarding model checking of RMCs \cite{EY-MC-12}
 with Theorem \ref{thm:qbd-termin-p-time}
and  Lemma \ref{lem:low-bound-probs-in-qbds}, and
with the key result by Brazdil, et. al. in \cite{BKK11}, 
which  establishes that non-zero {\em non-termination} probabilities
for a p1CA are also bounded away from zero by $1/2^{poly(|M|)}$.

As shown in \cite{CY95,EY-MC-12},  for probabilistic
model checking a {\em naive}
subset construction can be used (instead of
Safra's construction) to obtain from a BA, $\mathcal{B}$, 
a deterministic 
B\"{u}chi automaton, $D$,  such that $|D| = 2^{|B|}$.
(It need not be the case that $L(D) = L(B)$.)
One then constructs the ``product'' $M \otimes D$, of the
p1CA, $M$, with the deterministic B\"{u}chi automaton $D$.
A key observation is that
this ``product'' 
remains a p1CA.  
In terms of RMCs, p1CAs correspond to 1-box RMCs,
and the ``product'' of a 1-box RMC with a deterministic BA, $D$,
remains a 1-box RMC.

It was shown in \cite{EY-MC-12} that
given a ``product'' 
(1-box) RMC $M \otimes D$,  it is possible to construct a 
finite-state {\em conditioned summary chain},  $\mathcal{M'}$,
which is a finite state Markov chain and
whose transition probabilities are {\em rational expressions in positive
termination and non-termination probabilities} of the (1-box) RMC.
It is then possible to identify
in P-time certain bottom strongly connected components ${\mathcal{T}}$
of  $\mathcal{M'}$, 
such that the probability
$P_M(L(B))$ is equal to the probability that starting from a specific
initial state of $\mathcal{M'}$, a run eventually hits
a state in ${\mathcal{T}}$.

In this way, the model checking problem is boiled down to the 
problem of computing hitting probabilities in a {\em finite-state}
Markov chain whose transition probabilities are simple rational
expressions with numerators and
denominators that are products of coefficients in a p1CA
together with
 positive termination and non-termination probabilities
of a p1CA.

It is well known that non-zero hitting probabilities for a finite-state
Markov chain are the unique solution $(I-A)^{-1}b$, to 
a linear system of equations $x= Ax+b$, where the coefficients
in $A$ and $b$ come from the transition probabilities of the Markov
chain. 
The key remaining question is, {\em how well-conditioned
is this linear system of equations?}. In other words, what happens 
to its unique solution if we only
approximate the coefficients in $A$ and $b$ to within
a small error?
Now, the key is that applying Lemma \ref{lem:low-bound-probs-in-qbds} 
(which is from \cite{EWY-qbd-2010}),
and applying the key result in \cite{BKK11}, together shows that
both positive termination and positive non-termination probabilities 
of the product p1CA are bounded
away from $0$ by  $1/2^{poly(|M|,|D|)}$.   

Under these conditions, {\em exactly the same} known condition
number bounds from numerical analysis that were used  
in \cite{EWY-qbd-2010} 
namely
Theorem 17 of \cite{EWY-qbd-2010}, which
is a version of Theorem 2.1.2.3 of \cite{IsaKel66},
also establish that the linear system of equations
that one has to solve for 
hitting probabilities in the conditioned summary chain ${\mathcal{M}}'$
derived from a p1CA
are ``polynomially well-conditioned'', meaning that
approximating their non-zero coefficients within suitable $1/2^{poly}$ 
additive error yields
a linear system of equations whose unique solution is $\epsilon$-close to
the unique solution of the original system, for the chosen $\epsilon > 0$.
We omit a detailed elaboration in this conference paper.
\qed
\end{proof}

\section{Proof of Theorem \ref{thm:q-bounds}}

\label{sec:proof-of:thm:q-bounds}

\noindent {\bf Theorem  \ref{thm:q-bounds}.} 
{\em If $x=P(x)$ is a quadratic MPS in $n$ variables,  with LFP $q^* > 0$,  
and where $P(x)$ has rational coefficients and total encoding size $|P|$ bits, then
\begin{enumerate}

\item  \ \ $q^*_{\min} \geq 2^{-|P|(2^{n}-1)}$, \ \  \mbox{\rm and}\\

\item \ \ 
$q^*_{\max}  \leq 2^{2(n+1)(|P| + 2(n+1)\log(2n+2)) \cdot 5^{n}}$.
\end{enumerate}
}

\begin{proof} 
We first prove (1.), by lower bounding $q^*_{\min}$ in terms of the smallest constant $c_{\min}$ in $P(x)$.

 \begin{lemma} \label{lem:qmincmin} 
If $x=P(x)$ has LFP $q^* > 0$, and least term $c_{\min}$, then $q^*_{\min} \geq \min \{1, c_{\min} \}^{2^n -1}$. 
\end{lemma}
\begin{proof} 
We first observe that, since $q^* > 0$, and
there are $n$ variables, it must be the case that  $P^n(0) > 0$. 
To see this, for any $y \geq 0$, let us use $Z(y)$ to denote the set of zero coordinates of $y$.  
For any $k \geq 0$, $P^{k+l}(0) \geq P^k(0)$, for all $l \geq 0$, so $Z(P^{k+l}(0)) \subseteq Z(P^{k}(0))$.
Thus either $|Z(P^{k+1}(0))| = |Z(P^{k}(0))|$ or $|Z(P^{k+1}(0))| \leq |Z(P^{k}(0))| - 1$.
Now $|Z(0)| = n$ and $|Z(P^{k}(0))| \geq 0$ for all $k$, so there must be some least $0 \leq k \leq n$ 
such that $|Z(P^{k}(0))| = |Z(P^{k+1}(0))|$ and such that $Z(P^k(0)) = Z(P^{k+1}(0))$.

Note that, for any $y \geq 0$, $Z(P(y))$ depends only on $Z(y)$ and on $P(x)$, but not on the specific
values of non-zero coordinates of $y$.

So if for some $n \geq k \geq 0$, $Z(P^{k+1}(0)) = Z(P^{k}(0))$ then, by a simple induction $Z(P^{k+l}(0)) = Z(P^{k}(0))$ for all 
$l \geq 0$.  So we must have $Z(P^{k}(0)) = Z(P^{n}(0)) = Z(P^{n+l}(0))$, for all $l \geq 0$. 
Now $\lim_{m \rightarrow \infty} P^m(0) = q^*$. Now if $P^n(0)_i=0$, then $P^{n+l}(0)_i = 0$ for all $l \geq 0$, and so $q^*_i=0$. 
This contradicts our assumption that $q^* >0$. So $P^n(0) > 0$.

Let us use  $P^k(0)_{@}$ to denote the minimum value of any non-zero coordinate of $P^k(0)$.  
Firstly, $P(0) \not= 0$, i.e., there is some non-zero constant in some polynomial $P(x)_i$. 
Thus $P(0)_{@} \geq c_{\min}$.
We show by induction that for $k > 0$, $P^k(0)_{@} \geq \min \{1, c_{\min} \}^{2^k - 1}$. This is true for $k=0$. We assume 
that $P^k(0)_{@} \geq \min \{1, c_{\min} \}^{2^k - 1}$. If for some coordinate $i$, $P^{k+1}(0)_i = P(P^{k}(0))_i > 0$, 
there must be a term in $P(x)_i$ which is not zero in $P(P^{k}(0))_i$, this is either a constant $c$, 
or a linear term $cx_j$ with $P^k(0)_j > 0$, or a quadratic term $cx_jx_l$ with $P^k(0)_j > 0$ and  $P^k(0)_l > 0$.
In any of these 3 cases, this term is $\geq c_{\min} \min \{1, P^k(0)_{@} \}^2$. 
Since $P^k(0)_{@} \geq \min \{1, c_{\min} \}^{2^k - 1}$, we now have that $P^{k+1}(0)_{@} \geq c_{\min} (\min \{1, c_{\min} \}^{2^k - 1})^2 \geq  \min \{1, c_{\min} \}^{2^{k+1} - 1}$. 
So for all $k$, $P^k(0)_{@} \geq \min \{1, c_{\min} \}^{2^k - 1}$. In particular $P^n(0)_{@} \geq \min \{1, c_{\min} \}^{2^n - 1}$. But $P^n(0) > 0$ so $P^n(0)_{\min} \geq \min \{1, c_{\min} \}^{2^n - 1}$. 
We know $q^* \geq P^n(0)$, so $q^*_{\min} \geq \min \{1, c_{\min} \}^{2^n -1}$. 
\qed \end{proof}
To get our lower bound on $q^*_{\min}$ in terms of $|P|$ and $n$, we just note that clearly $c_{\min} \geq 2^{-|P|}$. 
This and Lemma \ref{lem:qmincmin} give the bound 
$q^*_{\min} \geq 2^{-|P|(2^{n}-1)}$ in part (1.) of the Theorem.

We now prove part (2.).
To prove the upper bound on $q^*_{\max}$, we need the following 
isolated root separation bound for systems of polynomial equations   
by Hansen et. al. \cite{miltersen-et-al-arxiv'12}:
\begin{theorem}{(Theorem 23 from \cite{miltersen-et-al-arxiv'12})}
  \label{th:isol-real-root-bd-full}
  Consider a polynomial system of equations
  \begin{equation}
    (\Sigma) \quad \quad
    g_1(x_1, \dots, x_n) = \cdots = g_m(x_1, \dots, x_n) = 0 \enspace,
  \label{eq:orig-system}
  \end{equation}
  with polynomials of degree at most $d$ and integer
  coefficients of magnitude at most $2^{\tau}$.

  Then, the coordinates of any {\em isolated} (in Euclidean topology) real 
  solutions of the system are
  real algebraic numbers of degree at most $(2d+1)^n$, and their
  defining polynomials have coefficients of magnitude at most 
  $2^{2n(\tau+4n\log(dm))(2d+1)^{n-1}}$.
  Also, if $\gamma_j =
  (\gamma_{j,1}, \cdots, \gamma_{j,n})$ is an isolated solution of $(\Sigma)$,
  then for any $i$, either
  \begin{equation}
    2^{-2n(\tau + 2n\log(dm))(2d+1)^{n-1}} < | \gamma_{j,i } |
    \quad \text{ or } \quad \gamma_{j,i} = 0 \enspace.
    \label{eq:isol-rr-lower-bd}
  \end{equation}
  Moreover, given coordinates of isolated solutions of two such systems, if they are not identical, they differ by at least
  \begin{equation}
    \sep(\Sigma) 
    \geq 2^{-3n(\tau + 2n\log(dm))(2d+1)^{2n-1} - \frac{1}{2}\log(n)}
    \label{eq:isol-rr-sep-bd}
    \enspace .
  \end{equation}
\end{theorem}

To apply  Theorem \ref{th:isol-real-root-bd-full}, we 
now establish that $q^*$ is
an isolated solution of an MPS with LFP $q^* > 0$. 
\begin{lemma} \label{q*-isolated} If $x=P(x)$ is a quadratic MPS with LFP $q^*> 0$, 
then $q^*$ is an isolated solution of the system of equations $x=P(x)$.\end{lemma}
\begin{proof}
Firstly, we consider strongly connected MPSs. These can be divided into two cases, linear strongly-connected MPSs, where $B(x) = B$  
is a constant matrix and $P(x)$ is affine, and nonlinear strongly-connected MPSs, where $B(x)$ 
is not a constant matrix and $P(x)$ is nonlinear.

For the linear case, the Jacobian is a constant $B(x) = B$, and  $x = P(x) = Bx + P(0)$.  
We know that $\rho(B(q^*)) \leq 1$ from Corollary \ref{cor:perron-cone},
and thus since $B = B(0) = B(q^*)$, from Lemma \ref{lem:spectr-at-lfp-is-at-most-1},
we know that $\rho(B) < 1$, and thus $(I-B)$ is non-singular,
and there is a unique solution to $x=P(x)= Bx + P(0)$, namely
$q^*=(I-B)^{-1} P(0)$. Being unique, this solution is isolated.

Now suppose, for contradiction, that $x=P(x)$ is a non-linear strongly-connected quadratic MPS but that 
$q^* > 0$ is not an isolated solution to $x=P(x)$. 
Because $q^*$ is not isolated, there is another fixed-point $q$ with $\|q^* - q\|_\infty \leq q^*_{\min}$ and $q \not= q^*$. 
Then $q \geq 0$ and, since $q^*$ is the least non-negative fixed-point,  $q \geq q^*$. From Lemma \ref{lem:restate-lem-3-3}
we have:
$$P(q) - P(q^*) = B(\frac{1}{2}(q^*+q))(q - q^*)$$
Because $q^*$ and $q$ are fixed points
 $$q - q^* = B(\frac{1}{2}(q^*+q))(q - q^*)$$
 Lemma \ref{coneratio} now yields that since $q -q^* \geq 0$ but $q -q^*\not= 0$ and $B(\frac{1}{2}(q^*+q))$ is irreducible, $q > q^*$. 
Thus $q - q^* > 0$ is a positive eigenvector of the irreducible matrix $B(\frac{1}{2}(q^*+q))$ associated with eigenvalue 1, 
thus $\rho(B(\frac{1}{2}(q^*+q))) = 1$ by Lemma \ref{lem:irr-unique-positive-eigenvector}.
 
We now again invoke the assumption of non-isolation of $q^*$,
which implies there is a vector $q' \neq q^*$ 
such that $q' = P(q')$ and $\|q^* - q'\|_\infty \leq \min \{q^*_{\min}, \frac{1}{2}(q - q^*)_{\min} \}$.
 By the same reasoning as above, we have that $q' > q^*$ and $\rho(B(\frac{1}{2}(q^*+q'))) = 1$. 
But now the condition $\|q^* - q'\|_\infty \leq \frac{1}{2}(q - q^*)_{\min}$ yields that 
$q' \leq q^* + \frac{1}{2}(q-q^*) < q$.  We thus also have that $\frac{1}{2}(q^*+q) > \frac{1}{2}(q^*+q')$,
and because $B(x)$ is non-constant and monotone in $x$, we have $B(\frac{1}{2}(q^*+q)) \geq B(\frac{1}{2}(q^*+q'))$ and
$B(\frac{1}{2}(q^*+q)) \not= B(\frac{1}{2}(q^*+q'))$. 
However, $\rho(B(\frac{1}{2}(q^*+q))) = 1 = \rho(B(\frac{1}{2}(q^*+q')))$. 
This contradicts Lemma \ref{lem:irr-non-neg-spec-lower}. So $q^*$ is also isolated in this case.

This establishes that for all strongly-connected MPSs, with LFP $q^* > 0$,  $q^*$ is isolated.

Now suppose that $x=P(x)$ is not strongly-connected. 
For each SCC $S$ of $x=P(x)$, the MPS $x_S=P_S(x_S, q^*_{D(S)})$ is strongly connected, 
so its LFP $q^*_S$ is an isolated solution of $x_S=P_S(x_S, q^*_{D(S)})$. 
That is, there is an $\epsilon_S > 0$ such that if $q_S$ has $\|q_S - q^*_S\| \leq \epsilon_S$ and 
$q_S=P_S(q_S, q^*_{D(S)})$, then $q_S = q^*_S$. Now take 
$\epsilon = \min_S \{ \epsilon_S \}$. We claim that if $\|q - q^*\|_\infty \leq \epsilon$ 
and $P(q) = q$, then $q=q^*$. We can show this by induction on the depth of strongly-connected components. 
If $S$ is a bottom strongly-connected component, then $q_S$ has $\|q_S - q^*_S\|_\infty \leq \epsilon \leq \epsilon_S$ and $q_S = P_S(q_S)$. So $q_S = q^*_S$.
If $S$ is a SCC and for all variables $D(S)$  that variables in $S$ depend on, directly or indirectly, $q_{D(S)} = q^*_{D(S)}$, then $q_S$ has 
$q_S = P_S(q_S, q_{D(S)}) = P_S(q_S, q^*_{D(S)})$. But this and $\|q_S - q^*_S\|_\infty \leq \epsilon \leq \epsilon_S$ are enough to establish 
$q_S = q^*_S$. This completes the induction showing that $q = q^*$. 
So $q^*$ is isolated solution for any MPS with LFP $q^* > 0$. \qed \end{proof}

For each $x_i$, let $d_i$ be the product of the denominators of  all coefficients of 
$P(x)_i$. Then $d_ix = d_iP(x)_i$ clearly has 
integer coefficients which are no larger than $2^{|P|}$.
Suppose $x = P(x)$ has LFP $q^* > 0$, and suppose 
that coordinate $k$ is the maximum coordinate of $q^*$, i.e., that $q^*_k = q^*_{\max}$. 
Now consider the system of $n+1$ polynomial equations,  in $n+1$ variables
(with an additional variable $y$), given by:

\begin{equation}
 d_ix_i = d_iP(x)_i \  , \ \ \mbox{and} \ \ \mbox{for all} \ i \in \{1,\ldots,n\}; \ \ \mbox{and}
 \ \ x_k \: y = 1  \ .
\label{eq:system-to-be-applied-on-iso}
\end{equation}

 Lemma \ref{q*-isolated} tells us that $q^* > 0$ is an isolated solution of $x=P(x)$. 
If $z \in \real^n$ is any solution vector for $x=P(x)$, there is a unique $w \in \real$ such that 
$x:=z$ and $y :=w$ 
forms a solution to the equations (\ref{eq:system-to-be-applied-on-iso});
namely let $w = \frac{1}{z_k}$.
So, letting $x := q^*$, and letting $y := \frac{1}{q^*_k}$ for all $i$, gives us
an isolated solution of the equations (\ref{eq:system-to-be-applied-on-iso}). 
 We can now apply Theorem \ref{th:isol-real-root-bd-full} to the system
(\ref{eq:system-to-be-applied-on-iso}).  For $y =\frac{1}{q^*_k}$, 
equation (\ref{eq:isol-rr-lower-bd}) in Theorem \ref{th:isol-real-root-bd-full} 
says that
  \begin{equation*}
    2^{-2(n+1)(|P| + 2(n+1)\log(2n+2))5^{n}} < \frac{1}{q^*_k}
    \quad \text{ or } \quad \frac{1}{q^*_k} = 0 \enspace.
  \end{equation*}
Since $q^* > 0$, clearly $\frac{1}{q^*_k} \neq 0$, so $\frac{1}{q^*_{\max}} = \frac{1}{q^*_k} > 
2^{-2(n+1)(|P| + 2(n+1)\log(2n+2))5^{n}}$. So 

\begin{equation}
q^*_{\max}  < 2^{2(n+1)(|P| + 2(n+1)\log(2n+2))5^{n}}.
\label{eq:upper-bound-on-qmax-tighter}
\end{equation}
\qed
  \end{proof}

\subsection{How good are our upper bounds for R-DNM on MPSs?}

\label{sec:app-how-good}

We prove in this paper upper bounds on the number of iterations
required by R-DNM to converge to within additive 
error $\epsilon > 0$ of the LFP
$q^*$ for an arbitrary MPS $x=P(x)$.

We now discuss some important parameters 
of the problem in which our upper bounds can
not be improved substantially.

To begin with, our upper bounds for the
number of iterations required contain a term of the form 
$2^d \log \frac{1}{\epsilon}$.
Here $d$ denotes the nesting depth of SCCs
in the dependency graph $G_P$ of the input MPS, $x=P(x)$.

It was already pointed out in \cite{lfppoly} (Section 7)
that such a term is a lower bound using Newton's method on MPSs,
even for exact Newton's method (whether decomposed or not),
even for rather simple MPSs.
\cite{lfppoly} provided a family of simple examples entailing the lower bound.
Indeed, consider the following MPS, $x = P(x)$, which is a simpler
variant of bad MPSs noted in \cite{lfppoly}.
The MPS has $n+1$ variables, $x_0, \ldots, x_n$.
The equations are:
\begin{eqnarray}
x_i & = & \frac{1}{2} x_i^2 + \frac{1}{2} x_{i-1} \ \ ,  \ \ 
\mbox{for all $i \in \{1,\ldots,n\}$}  \nonumber \\
 x_0 & = & \frac{1}{2} x_0^2 + \frac{1}{2} \nonumber
\end{eqnarray}
The LFP of this MPS is $q^* = {\mathbf 1}$,
and it captures the termination probabilities
of a (rather simple) stochastic context-free grammar, pBPA, 
or 1-exit Recursive Markov chain.  
Note that the encoding size of this MPS is $|P| = O(n)$.

As observed in \cite{rmc}, {\em exact} Newton's method,
starting from $x^{(0)} := 0$, on the
univariate equation 
$x_0 = \frac{1}{2} x_0^2 + \frac{1}{2}$ gains exactly one bit
of precision per iteration.  In other words,  if $x^{(k)}$
denotes the $k$'th iterate, then  $1- x^{(k)} = 2^{-k}$.

Suppose we perform $m$ iterations of exact NM on
the bottom SCC, $x_0 = \frac{1}{2} x_0^2 + \frac{1}{2}$,
and suppose that by doing 
so we obtain an appoximation $q'_0 = 1- a_0$, where $a_0 = 2^{-m}$.
Plugging the approximation $q'_0$ into the next higher SCC,
the equation for $x_1$ becomes 
$x_1  =  \frac{1}{2} x_1^2 + \frac{1}{2} q'_0$.
For the rest of the argument we do not need to appeal 
to Newton iterations: even {\em exact} computation
of the LFPs for the remaining SCCs will yield bad
approximations overall unless $1-q'_0 \leq \frac{1}{2^{2^n}}$ 
(showing that the system of equations
is terribly {\em ill-conditioned}).

Indeed, by induction on $i \geq 0$, 
suppose that the value obtained for LFP of $x_i$
is $q'_i = (1- a_i)$.
Then
after plugging in $q'_i$ in place of $x_i$
in the SCC for $x_{i+1}$, the
adjusted  LFP, $q'_{i+1}$, of the next higher SCC:
$x_{i+1}  = (1/2) (x_1)^2 + (1/2) (1 - a_{i})$,
becomes $q'_{i+1} = 1- \sqrt{a_i}$.
Thus, by induction on depth, the adjusted LFP of
$x_n$ becomes $q'_n =  1-  a_0^{2^{-n}}$.
But $a_0 = 2^{-m}$.  
Thus $q'_n = 1-  2^{-m2^{-n}}$.

We would like to have error $1- q'_n = 2^{-m2^{-n}} \leq \epsilon$.
Taking logs, we get that we must perform at least
$m \geq 2^n \log \frac{1}{\epsilon}$ NM iterations on the bottom SCC 
alone.

Note that $n$ here is also the depth $d$ of SCCs in this example.

Other terms in our upper bounds on
the number of iterations required to compute the LFP 
of a general MPS
are $\log \frac{1}{q^*_{\min}}$, and $\log {q^*_{\max}}$.
Simple ``repeated squaring''
MPSs, with $x_i = x_{i-1}^2$, $x_0 = \{ \frac{1}{2}$ or $2 \}$,   
show that we can have $q^*_{\min} \leq \frac{1}{2^{2^n}}$,
and $q^*_{\max} \geq 2^{2^n}$, where
$n$ is the number of variables.
In Theorem \ref{thm:q-bounds} we give explicit lower bounds on $q^*_{\min}$ and
explicit upper bounds on $q^*_{\max}$, 
in terms of $|P|$ and $n$,  showing
that linear-double-exponential dependence on $n$ 
is indeed the worst case possible.

However, it should be noted that 
the worst-case bounds on $q^*_{\min}$ and $q^*_{\max}$
are not representative of many important families of MPSs.
In particular, note that MPSs corresponding
to termination probabilities must have $q^*_{\max} \leq 1$.
Furthermore, for a number of classes of probabilistic systems
we can prove bounds of the form $\log \frac{1}{q^*_{\min}} \leq poly(|P|)$.
Indeed, for MPSs corresponding to   QBDs and probabilistic 1-counter
automata, such a bound was proved in \cite{EWY-qbd-2010}.

If the family of MPSs happens to have 
$\log \frac{1}{q^*_{\min}}, \log q^*_{\max}  \leq poly(|P|)$, then
our upper bounds show that the total number of iterations of R-DNM
needed is only exponential in $d$, the depth of SCCs,
and thus if $d \leq \log |P|$, then for such MPSs 
R-DNM runs {\em in P-time in the encoding size of the input, $|P|$ and 
$\log \frac{1}{\epsilon}$, in the standard Turing model
of computation}, to compute an approximation to the LFP $q^*$, within
additive error $\epsilon > 0$.

It should be noted that 
for the case of {\em strongly connected MPSs} only,
and only for {\em Exact Newton's Method}, without rounding, 
\cite{lfppoly} obtained comparable
result to ours in terms of worst-case dependence on
$\log \frac{1}{q^*_{\min}}$ and $\log q^*_{\max}$,\footnote{Technically,
their bounds are with respect to {\em relative error}, and their
bounds for strongly connected MPSs do not depend at all on $q^*_{\max}$, but of course 
if $q^*_{\max}$ is large, 
then in order to obtain absolute (additive) error $\epsilon > 0$,
the relative error required is $\epsilon' = \frac{\epsilon}{q^*_{\max}}$, 
and since their bounds depend on $\log \frac{1}{\epsilon'}$
they depend (indirectly) on $\log{q^*_{\max}}$, with the same
magnitude as ours.}
However, 
in \cite{lfppoly} they did not obtain any constructive bounds in terms
of $|P|$, $q^*_{\min}$ or $q^*_{\max}$  for MPSs that are not
strongly connected,  nor 
did they obtain any results for rounded versions of Newton's method. 
Using exact Newton's method of course
entails the assumption of a unit-cost arithmetic model of computation,
rather than the Turing model.

\end{document}